\newif\ifconf
\conftrue 
\conffalse

\documentclass[letterpaper,10pt]{article}

\usepackage{typearea}
\typearea{15}

\usepackage[compact]{titlesec}

\usepackage[margin=1.5in]{geometry}
\usepackage{appendix}
\usepackage{amsthm,amsfonts,amsmath,amssymb,dsfont}
\usepackage{cite,url,balance}
\usepackage{algorithm,algorithmic}
\usepackage{graphicx,psfrag}
\usepackage[usenames,dvipsnames]{color}
\usepackage{array,multirow}

\newenvironment{proofof}[1]{\noindent{\bf Proof of {#1}:}}
{\qed

}

\newcommand{\confoption}[2]{{\ifconf #1 \else #2 \fi}}

\newcounter{note}[section]

\newtheorem{theorem}{Theorem}
\newtheorem{definition}{Definition}

\newtheorem{cor}[theorem]{Corollary}
\newtheorem{lemma}[theorem]{Lemma}

\def\rank{{\rm rank }} 
\def\tr{{\rm tr}} 
\def\E{\mathbb{E}} 
\def\Pr{{\rm Pr}} 

\def\R{{\mathds{R}}} 

\newcommand{\junk}[1]{}

\def\b0{{\bf 0}}

\renewcommand{\S}{\mathbb{S}}

\DeclareMathOperator{\conv}{conv}
\DeclareMathOperator{\disc}{disc}
\DeclareMathOperator{\herdisc}{herdisc}
\DeclareMathOperator{\vecdisc}{vecdisc}
\DeclareMathOperator{\hvecdisc}{hvecdisc}

\DeclareMathOperator{\specLB}{specLB}
\DeclareMathOperator{\diag}{diag}

\DeclareMathOperator{\range}{range}

\newcommand{\cut}[1]{}

\begin{document}
\title{Approximating Hereditary Discrepancy via Small Width Ellipsoids}
\author{Aleksandar Nikolov\\Rutgers University \and Kunal
  Talwar\\Microsoft Research}
\date{}
\maketitle
\begin{abstract}
The {\em Discrepancy} of a hypergraph is the minimum attainable value, over  two-colorings of its vertices, of the maximum absolute imbalance of any hyperedge. The {\em Hereditary Discrepancy} of a hypergraph, defined as the maximum discrepancy of a restriction of the hypergraph to a subset of its vertices, is a measure of its complexity. Lov\'{a}sz, Spencer and Vesztergombi (1986) related the natural extension of this quantity to matrices  to rounding algorithms for linear programs, and gave a determinant based lower bound on the hereditary discrepancy. Matou\v{s}ek (2011) showed that this bound is tight up to a polylogarithmic factor, leaving open the question of actually computing this bound. Recent work by Nikolov, Talwar and Zhang (2013) showed a polynomial time $\tilde{O}(\log^3 n)$-approximation to hereditary discrepancy, as a by-product of their work in differential privacy. In this paper, we give a direct simple $O(\log^{3/2} n)$-approximation algorithm for this problem. We show that up to this approximation factor, the hereditary discrepancy of a matrix $A$ is characterized by the optimal value of  simple geometric convex program that seeks to minimize the largest $\ell_{\infty}$ norm of any point in a ellipsoid containing the columns of $A$. This characterization promises to be a useful tool in discrepancy theory.
\end{abstract}
\thispagestyle{empty}
\setcounter{page}{0}
\newpage
\section{Introduction}

Discrepancy theory, in the broadest sense, studies the fundamental
limits to approximating a ``complex measure'' (i.e.~continuous, or
with large support) with a ``simple measure'' (i.e. counting measure,
or a measure with small support) with respect to a class of
``distinguishers''. A prototypical ``continuous discrepancy'' question
is how uniform can a set $P$ of $n$ points in the unit square $[0,
1)^2$ be, where uniformity is measured with respect to a class of
geometric shapes, e.g.~axis-aligned rectangles~\cite{schmidt}. A
prototypical ``discrete discrepancy'' question asks whether we can
color the $n$ vertices of a hypergraph of $O(n)$ edges with two
colors, red and blue, so that each edge has approximately the same
number of red vertices as blue vertices~\cite{spencer-six}. The two
kinds of questions are deeply related, and transference theorems
between different discrepancy measures are known~\cite{beck-sos}.

Questions related to discrepancy theory are raised throughout
mathematics, e.g.~number theory, Diophantine approximation, numerical
integration. Unsurprisingly, they also naturally appear in computer
science -- questions about the (im)possibility of approximating
continuous, average objects with discrete ones are central to
pseudorandomness, learning theory, communication complexity,
approximation algorithms, among others. For a beautiful survey of
applications of discrepancy theory to computer science, we refer the reader to
Chazelle's The Discrepancy Method~\cite{Chazelle}.

Despite discrepancy theory's many applications in computer science, we
have only recently began to understand the computational complexity of
measures of discrepancy themselves. In this work, we address the
problem of approximately computing hereditary discrepancy, one of the
fundamental discrepancy measures. Hereditary discrepancy is a robust
version of combinatorial discrepancy, which is the hypergraph coloring
problem mentioned above. More precisely, the combinatorial discrepancy
$\disc(\mathcal{H})$ of a hypergraph $\mathcal{H} = (H_1, \ldots,
H_m)$ on the vertices $[n] = \{1, \ldots, n\}$ is the minimum over
colorings $\chi:[n] \rightarrow \{-1, 1\}$ of the maximum
``imbalance'' over hyperedges $\max_{i = 1}^m{|\sum_{j \in
    H_i}{\chi(i)}|}$.  While relatively simple, $\disc(\mathcal{H})$
is a brittle quantity, which can make it intractable to
estimate.\confoption{}{We may wish to say that $\disc(\mathcal{H})$
measures the complexity of $\mathcal{H}$, but it can be $0$ for
intuitively complex $\mathcal{H}$ for trivial reasons. For example let
$(V,E)$ be a complex hypergraph all of whose sets have equal size,
such as ${[n] \choose n/2}$. Consider the hypergraph formed by taking
two identical copies of $(V, E)$, say $(V_1, E_1)$ and $(V_2, E_2)$
and defining the new hypergraph as $(V_1 \cup V_2, E'\triangleq\{e_1
\cup e_2: e_1 \in E_1, e_2 \in E_2\})$. By coloring $V_1$ as $+1$ and
$V_2$ as $-1$, we get discrepancy zero for each edge in $E'$, despite
the intuitive complexity of $E$.  }For this reason it is often more
convenient to work with the more robust \emph{hereditary
  discrepancy}. Hereditary discrepancy is the maximum discrepancy over
restricted hypergraphs, i.e. $\herdisc(\mathcal{H}) = \max_{W\subseteq
  [n]} \disc(\mathcal{H}|_W)$, where $\mathcal{H}|_W = (H_1 \cap W,
\ldots, H_m \cap W)$.\confoption{}{Notice, for example, that the hereditary
discrepancy of the above example is in fact $\Omega(n)$ -- a more
fitting measure of the complexity of the
hypergraph.} 

Discrepancy and hereditary discrepancy have natural generalizations to
matrices. The discrepancy of a matrix $A$ is equal to $\disc(A) =
\min_{x\in \{-1, 1\}^n}{\|Ax\|_\infty}$, and hereditary discrepancy is
equal to $\max_{S \subseteq [n]}{\disc(A|_S)}$, where $A|_S$ is the
submatrix of $A$ consisting of columns indexed by elements of
$S$. These quantities coincide with hypergraph discrepancy when
evaluated on the incidence matrix of the hypergraph, and are also
natural themselves. For example, a classical result of Lov\'{a}sz,
Spencer, and Vesztergombi~\cite{LSV} states that for any matrix $A$,
any vector $c \in [-1, 1]^n$ can be rounded to $x \in \{-1,1\}^n$ so
that $\|Ax-Ac\|_\infty \leq 2\herdisc(A)$. In the context of a linear
program, this means that we can round fractional solutions to integral
ones while still approximately satisfying the linear constraints defined by $A$. This
fact was recently used by Rothvo{\ss} to design an improved
approximation algorithm for bin packing~\cite{rothvoss-binpacking}.

The robustness of hereditary discrepancy in comparison with
discrepancy is evident in the hardness of approximating each of the
two measures. By an important result of Spencer~\cite{spencer-six},
whenever the number of edges in $\mathcal{H}$ is $m = O(n)$, the
discrepancy is at most $O(\sqrt{n})$. It turns our that it is $\mathsf{NP}$-hard to
distinguish between this worst-case upper bound and 
discrepancy zero~\cite{CNN}. By contrast, recently Nikolov, Talwar, and
Zhang gave a polylogarithmic approximation to hereditary
discrepancy~\cite{NTZ}. At first glance, this is surprising, because
hereditary discrepancy is a maximum over exponentially many
$\mathsf{NP}$-hard 
minimization problems. Thus, hereditary discrepancy is not even obviously in
$\mathsf{NP}$ (but is $\mathsf{NP}$-hard to approximate within a
factor of 2, see~\cite{AustrinGH13}). However, the structure and robustness of hereditary
discrepancy explain its more tractable nature. As one classical
illustration to this, we note that the hypergraphs with hereditary
discrepancy $1$ are exactly the hypergraphs with totally unimodular
incidence matrices~\cite{gh-h-tum}, and are recognizable by a
polynomial time algorithm~\cite{seymour-tum}.

\paragraph{Our Results and Techniques.} When approximating a function
$f$, we need to provide (nearly matching) upper bounds and lower
bounds on $f$. When we approximate $\mathsf{NP}$-optimization
problems, usually proving either the upper (for maximization problems)
or the lower bound (for minimization problems) is relatively
straightforward: it is given by a combinatorial lower (or upper) bound
or a convex relaxation. The challenge is to design a bound which is
nearly tight. On the other hand, in a max-min problem like hereditary
discrepancy, both upper and lower bounds are challenging to
prove. Nevertheless, a convex relaxation of \emph{discrepancy} still
turns out to be very useful. The relaxation, vector discrepancy, is
derived by relaxing the condition on the coloring $\chi:[n]\rightarrow
\{-1, 1\}$ to the weaker $\chi:[n] \rightarrow \S^{n-1}$, where
$\S^{n-1}$ is the unit sphere in $\R^n$. Then the vector discrepancy
$\vecdisc(\mathcal{H})$ is the minimum over such $\chi$ of $\max_{i =
  1}^m{\|\sum_{j \in H_i}{\chi(i)}\|_2}$. Similarly, the extension
$\vecdisc(A)$ to matrices $A$ is the minimum over $\chi:[n]
\rightarrow \S^{n-1}$ of $\max_{i =
  1}^m{\|\sum_{j=1}^n{A_{ij}\chi(i)}\|_2}$.  These quantities can be
efficiently approximated to within any prescribed accuracy by solving
a semidefinite program. The \emph{hereditary} vector discrepancy
$\hvecdisc(A)$ is defined analogously to hereditary discrepancy as the
maximum vector discrepancy over submatrices. Clearly, $\vecdisc(A) \leq \disc(A)$ for any matrix $A$.
There exist matrices $A$ with $m = O(n)$ for which $\vecdisc(A) = 0$
and $\disc(A) = \Omega(\sqrt{n})$\footnote{This is the case, for
  example, for the matrix which contains three copies of each column
  of a Hadamard matrix.}. Nevertheless, in a recent breakthrough,
Bansal showed that an upper bound on hereditary vector discrepancy is
useful in efficient discrepancy minimization.
\begin{theorem}[\cite{Bansal10}]\label{thm:bansal}
  Let $A$ be an $m$ by $n$ matrix with $\hvecdisc(A) \leq
  \lambda$. Then there exists a randomized polynomial time algorithm
  that computes $x \in \{-1, 1\}^n$ with
  discrepancy at most  $\|Ax\|_\infty \leq O(\log m) \cdot\lambda$.
\end{theorem}

Theorem~\ref{thm:bansal} implies that the gap between
$\hvecdisc$ and $\herdisc$ is at most logarithmic.
\begin{cor}\label{cor:bansal}
  For any $m \times n$ matrix $A$
  \begin{equation*}
    \hvecdisc(A) \leq \herdisc(A) \leq O(\log m)
    \hvecdisc(A). 
  \end{equation*}
\end{cor}
While $\vecdisc(A)$ can be approximated to
within any degree of accuracy in polynomial time, it is not clear if
$\hvecdisc(A)$ can be computed efficiently: notice that
hereditary  vector discrepancy is the maximum of the
objective functions of an exponential number of convex minimization
problems. Nevertheless, by using vector discrepancy we remove
one of the two quantifiers over exponentially large sets. 

In this paper we prove the following approximation result for
$\hvecdisc$. 
\begin{theorem}\label{thm:vecdisc-apx}
  There exists a polynomial time algorithm that approximates
  $\hvecdisc(A)$ within a factor of $O(\log m)$ for any $m\times n$
  matrix $A$. Moreover, the algorithm finds a submatrix $A|_S$ of $A$,
  such that $\hvecdisc(A) = O(\log m) \vecdisc(A|_S)$.
\end{theorem}

Theorem~\ref{thm:vecdisc-apx} follows from a geometric characterization of
hereditary vector discrepancy. We show that, up to a
factor of $O(\log m)$, $\hvecdisc(A)$ is equal to the
smallest value of $\|E\|_\infty$ over all ellipsoids that contain the
columns of $A$. Here, $\|E\|_\infty$ is just the maximum
$\ell_\infty^m$ norm of all points in $E$, or, equivalently, the
maximum width of $E$ in the directions of the standard basis vectors
$e_1, \ldots, e_m$. {\em A priori}, it is not clear how to relate this quantity in either direction to the $\hvecdisc(A)$, as it is not a fractional ``relaxation'' in the traditional sense. It is in fact non-trivial to prove either of the two inequalities relating the geometric quantity to $\hvecdisc(A)$.

Proving that this quantity is an upper bound on hereditary discrepancy relies on a recent result of Nikolov
that upper bounds the vector discrepancy of matrices with columns
bounded in Euclidean norm by $1$~\cite{komlos-sdp}. We need a slight
generalization of Nikolov's result that shows that the vector
discrepancy of such matrices can be bounded by $1$ \emph{in any
  direction}. We then transform linearly the containing ellipsoid $E$ to a
unit ball, so that Nikolov's result applies; because of the transformation,
we need to make sure that in the transformed space the vector
discrepancy is low in a set of directions different from the standard
basis. While, on the face of things, this argument only upper bounds
the vector discrepancy of $A$, it in fact also upper bounds the vector
discrepancy of \emph{any submatrix} as well, because if $E$ contains all
columns of $A$, it also contains all the columns of any submatrix of
$A$. This simple observation is crucial to the success of our arguments.

To show that the smallest value of $\|E\|_\infty$ over all containing
ellipsoids also gives a lower bound on hereditary vector discrepancy,
we analyze the \emph{convex dual} of the problem of finding containing
ellipsoids of small width and show that we can transform dual
certificates for this problem to dual certificates for vector
discrepancy of some submatrix of $A$. The dual of the problem of
minimizing $\|E\|_\infty$ for a matrix $A$ is a problem of maximizing
the \emph{nuclear norm} (i.e. the sum of singular values) over
re-weightings of the columns and rows of $A$. To get dual certificates for
vector discrepancy for some submatrix, we need to be able to extract a
submatrix with a large least singular value from a matrix of large
nuclear norm. We accomplish this using the \emph{restricted
  invertibility principle} of Bourgain and Tzafriri~\cite{bour-tza}: a
powerful theorem from functional analysis which states, roughly, that
any submatrix with many approximately equal singular values contains a
large well-conditioned submatrix. Using a constructive proof of the
theorem by Spielman and Srivastava~\cite{bt-constructive}, we can also find the well-conditioned submatrix in deterministic polynomial time;
this gives us a submatrix of $A$ on which hereditary vector
discrepancy is approximately maximized.

Theorem~\ref{thm:vecdisc-apx} immediately implies a $O(\log^2 m)$
approximation of $\herdisc$ via Bansal's theorem. However, we can
improve this bound to an $O(\log^{3/2} m)$ approximation.
\begin{theorem}\label{thm:herdisc-apx}
  There exists a polynomial time algorithm that approximates
  $\herdisc(A)$ within a factor of $O(\log^{3/2} m)$ for any $m \times
  n$ matrix $A$. Moreover, the algorithm finds a submatrix $A|_S$ of
  $A$, such that $\herdisc(A) \leq O(\log^{3/2} m) \vecdisc(A|_S)$. 
\end{theorem}
To prove Theorem~\ref{thm:herdisc-apx}, we lower bound hereditary vector
discrepancy as before, in order to lower bound hereditary discrepancy. However,
for the upper bound, rather than upper bounding vector discrepancy in
terms of $\|E\|_\infty$ for a containing ellipsoid, and then upper
bounding discrepancy in terms of vector discrepancy, we directly upper
bound discrepancy in terms of $\|E\|_\infty$. For this purpose, we use
another discrepancy bound -- this time a theorem due to
Banaszczyk~\cite{bana} that shows that for any convex body $K$ of large Gaussian
volume, and a matrix $A$ with columns of at most unit Euclidean norm,
there exists a $x \in \{-1, 1\}^n$ such that $Ax \in CK$ for a
constant $C$. We use this theorem analogously to the way we used
Nikolov's theorem: we linearly transform $E$ to the unit ball, and we specify a
body $K$ such that if some $\pm 1$ combination of the columns of $A$
is in $K$ after the transformation, then in the preimage
the combination is in an infinity ball scaled by $O(\sqrt{\log m})$. 

After a preliminary version of this paper was made available,
Matou\v{s}ek~\cite{matousek-says} has shown that our analysis of both
the upper and the lower bound on $\herdisc(A)$ in terms of the minimum
of $\|E\|_\infty$ over containing ellipsoids $E$ is tight. He also used our characterization of $\herdisc(A)$ in terms of the minimum
of $\|E\|_\infty$ and our analysis of the dual to give new proofs of
classical and new results in discrepancy theory.

\paragraph{Comparison with Related Works.}

Lov\'{a}sz, Spencer and Vesztergombi~\cite{LSV} defined a determinant
based lower bound on the hereditary discrepancy of a
matrix. Matou\v{s}ek~\cite{Matousek11} showed that this lower bound is
tight up to $O(\log^{3/2} m)$. These results did not immediately yield
an approximation algorithm for hereditary discrepancy, as the determinant lower
bound is a maximum over exponentially many quantities and not known to
be efficiently computable. 

Nikolov, Talwar and Zhang~\cite{NTZ} recently studied hereditary
discrepancy as a tool for designing near optimal differentially
private mechanisms for linear queries, and as a by-product, derived an
$\tilde{O}(\log^3 n)$-approximation algorithm for hereditary
discrepancy, where the $\tilde{O}$ notation hides sub-logarithmic
factors. Small width containing ellipsoids were implicit in their
work. The current paper is the first that explicitly considers this
natural geometric object in the context of discrepancy. While the
proof of the upper bound on discrepancy in~\cite{NTZ} was via a
connection between discrepancy and differential privacy due to Nikolov
and Muthukrishnan~\cite{halfspaces}, here we give a \emph{tight} and
\emph{direct} argument using results of Nikolov~\cite{komlos-sdp} and
Banaszczyk~\cite{bana} on the Koml\'{o}s problem. The arguments via
differential privacy cannot give the tight relationship between the
minimum width of a containing ellipsoid and hereditary discrepancy:
they necessarily lose a logarithmic factor, because the relationship
between discrepancy and privacy is itself not tight. Moreover, our
arguments are simpler and more transparent. The proof of our lower
bound is also more natural: we relate the duals of the two convex
optimization problems under consideration, i.e.~the problem of
minimizing vector discrepancy, and the problem of minimizing the width
of a containing ellipsoid. Via this approach we arrive at a new
discrepancy lower bound, which is at least as strong as the
determinant lower bound (up to a logarithmic factor), is tight with
respect to hereditary discrepancy up to the same asymptotic factor of
$O(\log^{3/2} m)$, and is efficiently computable. We believe our lower
bound will have future applications in discrepancy theory.

A separate but related line of recent works~\cite{Bansal10,lovettmeka,Rothvoss14-giann} gives constructive versions of existence proofs in discrepancy theory. 


\section{Preliminaries}

We start by introducing some basic notation.

For a $m \times n$ matrix $A$ and a set $S \subseteq [n]$, we denote
by $A|_S$ the submatrix of $A$ consisting of those columns of $A$
indexed by elements of $S$. $\mathcal{P}_k$ is the set of orthogonal
projection matrices onto $k$-dimensional subspaces of $\R^m$.  We use
$\range(A)$ for the range, i.e. the span of the columns, of $A$.  By
$\sigma_{\min}(A)$ and $\sigma_{\max}(A)$ we denote, respectively, the
smallest and largest singular value of $A$. I.e., $\sigma_{\min}(A) =
\min_{x: \|x\|_2 = 1}{\|Ax\|_2}$ and $\sigma_{\max}(A) = \max_{x:
  \|x\|_2 = 1}{\|Ax\|_2}$. In general, we use $\sigma_i$ for the
$i$-th largest singular value of $A$.

By $X \succeq 0$ ($X\succ 0$) we denote that $X$ is a positive
semidefinite (resp.~positive definite) matrix,
and by $X \preceq Y$ that $Y - X \succeq 0$. 

\confoption{}{Recall that for a block
matrix \[X = \left(\begin{array}{cc}A &B\\B^T &C\end{array}\right),\]
the \emph{Schur complement} of an invertible block $C$ in $X$ is $A -
B^TC^{-1}B$. When $C \succ 0$, $X \succeq 0$ if and only if $A -
B^TC^{-1}B \succeq 0$.}

For a positive semidefinite (PSD) matrix $X \succeq 0$, we denote by
$X^{1/2}$ the principal square root of $X$, i.e. the matrix $Y \succeq
0$ such that $Y^2 = X$.

\subsection{Matrix Norms and Restricted Invertibility}

\junk{For a symmetric matrix $M$, we denote by $\lambda_i(M)$ the
$i$-th largest eigenvalue of $M$. 

We recall the variational (minimax)
characterization of eigenvalues:
\begin{equation*}
  \lambda_i(M) = \max_{\mathcal{V}: \dim \mathcal{V} = i}\min_{y \in
    \mathcal{V}:\|x\|_2 = 1}{x^TMx}.
\end{equation*}

By $\sigma_{\min}(A)$ and $\sigma_{\max}(A)$ we denote, respectively,
the smallest and largest singular value of $A$. I.e.,
$\sigma_{\min}(A) = \min_{x: \|x\|_2 = 1}{\|Ax\|_2}$ and
$\sigma_{\max}(A) = \max_{x: \|x\|_2 = 1}{\|Ax\|_2}$. In general,
$\sigma_i(A) = \lambda_i(A^TA) = \lambda_i(AA^T)$. 
By the spectral theorem, any real $m \times n$ matrix $A$ of
rank $r$ can be factorized as $A = U\Sigma V^T$, where $U$ and $V$ are
respectively $m \times r$ and $n \times r$ unitary matrices
(i.e. $U^TU = V^TV = I_r$) and $\Sigma$ is a $r \times r$ diagonal matrix with
$\Sigma_{ii} = \sigma_i > 0$. }

The Schatten 1-norm of a matrix $A$,
also known as the trace norm or the nuclear norm, is equal to
$\|A\|_{S_1} = \sum_{i}{\sigma_i(A)} = \tr((AA^T)^{1/2})$.

For a matrix $A$, we denote by $\|A\|_2=\sigma_{\max}(A)$ the spectral
norm of $A$ and $\|A\|_{HS} = \sqrt{\sum_i \sigma_i^2(A)} =
\sqrt{\sum_{i,j} a_{i,j}^2}$ the Hilbert-Schmidt (or Frobenius) norm
of $A$. We use $\|A\|_{1 \rightarrow 2}$ for the maximum Euclidean length
of the columns of the matrix $A = (a_i)_{i = 1}^n$, i.e.~$\|A\|_{1 \rightarrow 2} = \max_{x:
  \|x\|_1 = 1}{\|Ax\|_2} = \max_{i \in [n]}{\|A_i\|_2}$.

A matrix $A$ trivially contains an invertible submatrix of $k$ columns
as long as $k \leq \rank(A)$. An important result of Bourgain and
Tzafriri~\cite{bour-tza} (later strengthened by
Vershynin~\cite{vershynin}, and Spielman and
Srivastava~\cite{bt-constructive}) shows that when $k$ is strictly
less than the robust rank $\|A\|_{HS}^2/\|A\|_2^2$ of $A$, we can find
$k$ columns of $A$ that form a \emph{well-invertible} submatrix. This
result is usually called the \emph{restricted invertibility
  principle}. Next we state a weighted version of it, which can be
proved by slightly modifying the proof of Spielman and
Srivastava~\cite{bt-constructive}. Rather than describe the
modification, we give a reduction of the weighted version to the
standard statement in Appendix~\ref{app:weighted-rip}.

\begin{theorem}\label{thm:bt}
  Let $\epsilon > 0$, let $A$ be an $m$ by $n$ real matrix, and let
  $Q$ be a diagonal matrix such that $Q \succeq 0$ and $\tr(Q) =
  1$. For any integer $k$  such that $k \leq \epsilon^2
  \frac{\|AQ^{1/2}\|_{HS}^2}{\|AQ^{1/2}\|_2^2}$ there exists a subset $S
  \subseteq [n]$ of size $S = k$ such that 
  $\sigma_{\min}(A|_S)^2 \geq (1-\epsilon)^2\|AQ^{1/2}\|_{HS}^2$. 
  Moreover, $S$ can be computed in deterministic polynomial time. 
\end{theorem}

\subsection{Geometry}


Let $\conv\{a_1, \ldots a_n\}$ be the convex hull of the vectors $
a_1, \ldots, a_n$. 

A \emph{convex body} is a convex compact subset of $\R^m$. For a
convex body $K \subseteq \R^m$, the \emph{polar body} $K^\circ$ is
defined by $K^\circ = \{y: \langle y, x \rangle \leq 1~\forall x \in
K\}$. A basic fact about polar bodies  is that for any
two convex bodies $K$ and $L$, $K \subseteq L \Leftrightarrow L^\circ
\subseteq K^\circ$. Moreover, a symmetric convex body $K$ and its polar body are
dual to each other, in the sense that $(K^\circ)^\circ = K$.

A convex body $K$ is \emph{(centrally) symmetric} if $-K = K$. The
\emph{Minkowski norm} $\|x\|_K$ induced by a symmetric convex body $K$
is defined as $\|x\|_K \triangleq \min\{r \in \R: x \in rK\}$. The
Minkowski norm induced by the polar body $K^\circ$ of $K$ is the
\emph{dual norm} of $\|x\|_K$ and also has the form $\|y\|_{K^\circ} =
\max_{x \in K}{\langle x, y\rangle}$. It follows that we can also
write $\|x\|_K$ as $\|x\|_K = \max_{y \in K^\circ}{\langle x, y
  \rangle}$. For a vector $y$ of unit Euclidean length,
$\|y\|_{K^\circ}$ is the \emph{width} of $K$ in the direction of $y$,
i.e.~half the Euclidean distance between the two supporting hyperplanes of $K$
orthogonal to $y$. For symmetric body $K$, we denote by $\|K\| = \max_{x \in  K}{\|x\|}$ the diameter of $K$ under the norm $\|\cdot\|$.

Of special interest are the $\ell_p^m$ norms, defined for any $p \geq
1$ and any $x \in \R^m$ by $\|x\|_p = \left(\sum_{i =
    1}^m{|x|^p}\right)^{1/p}$. The $\ell_\infty^m$ norm is defined for
as $\|x\|_\infty = \max_{i = 1}^m{|x_i|}$. The norms
$\ell_p^m$ and $\ell_q^m$ are dual if and only if $\frac{1}{p} +
\frac{1}{q} = 1$, and $\ell_1^m$ is dual to $\ell_\infty^m$. We denote
the unit ball of the $\ell_p^m$ norm by $B_p^m = \{x: \|x\|_p \leq
1\}$. As with the unit ball of any norm, $B_p^m$ is convex and
centrally symmetric for $p \in [1, \infty]$.

An \emph{ellipsoid} in $\R^m$ is the image of the ball $B_2^m$ under
an affine map. All ellipsoids we consider are symmetric, and
therefore, are equal to an image $F B_2^m$ of the ball $B_2^m$ under a
linear map $F$. A full dimensional ellipsoid $E = FB_2^d$ can be
equivalently defined as $E = \{x: x^T(FF^T)^{-1}x \leq 1\}$.  The
polar body of a symmetric ellipsoid $E = F B_2^d$ is the ellipsoid
$E^\circ = \{x: x^TFF^Tx \leq 1\}$. It follows that for $E = FB_2^m$
and for any $x$, $\|x\|_E = \sqrt{x^T(FF^T)^{-1}x}$ and for any $y$,
$\|y\|_{E^\circ} = \sqrt{y^T(FF^T)y}$.

\confoption{}{\subsection{Convex Duality}

Assume we are given the following optimization problem:
\begin{align}
  &\text{Minimize } f_0(x)\label{eq:general-obj}\\
  &\text{s.t.}\notag\\
  &\forall 1\leq i \leq m: f_i(x) \leq 0.\label{eq:general-constr}
\end{align}
The Lagrange dual function associated with
\eqref{eq:general-obj}--\eqref{eq:general-constr} is defined as 
  $g(y) = \inf_x f_0(x) + \sum_{i = 1}^m{y_if_i(x)}$,
  where the infimum is over the intersection of the domains of
  $f_1,\ldots,\ldots f_m$, and $y \in \R^m$, $y \geq 0$. Since $g(y)$
  is the infimum of affine functions, it is a concave
  function. Moreover, $g$ is upper semi-continuous, and therefore
  continuous over the convex set $\{y: g(y) > -\infty\}$. 

  For any $x$ which is feasible for
  \eqref{eq:general-obj}--\eqref{eq:general-constr}, and any $y \geq
  0$, $g(y) \leq f_0(x)$. This fact is known as \emph{weak
    duality}. The \emph{Lagrange dual problem} is defined as
\begin{align}
  &\text{Maximize } g(y)
  \text{ s.t. }
  y \geq 0.\label{eq:L-dual}
\end{align}
\emph{Strong duality} holds when the optimal value of
\eqref{eq:L-dual} equals  the optimal
value of \eqref{eq:general-obj}--\eqref{eq:general-constr}. Slater's
condition is a commonly used sufficient condition for strong
duality. We state it next.

\begin{theorem}[Slater's Condition]\label{thm:slater}
  Assume $f_0, \ldots, f_m$ in the problem
  \eqref{eq:general-obj}--\eqref{eq:general-constr} are convex
  functions over their respective domains, and for some $k \geq 0$,
  $f_1, \ldots, f_k$ are affine functions. Let there be a point $x$ in
  the relative interior of the domains of $f_0, \ldots, f_m$, so that
  $f_i(x) \leq 0$ for $1 \leq i \leq k$ and $f_j(x) < 0$ for $k+1 \leq
  j \leq m$. Then the minimum of
  \eqref{eq:general-obj}--\eqref{eq:general-constr} equals the maximum
  of \eqref{eq:L-dual}, and the maximum of
  \eqref{eq:L-dual} is achieved if it is
  finite.
\end{theorem}

For more information on convex programming and duality, we refer the
reader to the book by Boyd and Vandenberghe.
}

\subsection{Hereditary Discrepancy and Relaxations}

\confoption{Hereditary discrepancy and (hereditary) vector discrepancy
  are defined as in the Introduction. We note that $\vecdisc(A)$ can
  also be equivalently characterized as the minimum of $\max_{i =
    1}^m{\sqrt{e_i^T(AXA^T)e_i}}$ over matrices $X \succeq 0$, where
  $e_i$ is the $i$-th standard basis vector. }{For a
  $m \times n$ matrix $A$, discrepancy and hereditary discrepancy are
  defined as
\begin{align*}
  &\disc(A) \triangleq \min_{x \in \{\pm 1\}^n}{\|Ax\|_\infty}
  &\herdisc(A) \triangleq \max_{S \subseteq [n]}{\disc(A|_S)}.
\end{align*}

Vector discrepancy is a convex relaxation of discrepancy in which one
can assign arbitrary unit vectors rather than $\pm 1$ to the columns
of $A$. Formally, let $\S^{n-1}$ be the unit sphere in $\R^n$, and
define
\begin{equation*}
  \vecdisc(A) \triangleq \min_{u_1, \ldots, u_n \in \S^{n-1}} \max_{i = 1}^m
  \left\| \sum_{j = 1}^n{A_{ij} u_j} \right\|_2
\end{equation*}

From a computational complexity perspective, the important property of
vector discrepancy $\vecdisc(A)$ is that it can be approximated to within an
arbitrarily small additive constant $\epsilon$ in time polynomial in
$\log \frac{1}{\epsilon}$ by (approximately) solving a semidefinite
program. }The following lower bound on vector discrepancy follows by weak convex
duality for semidefinite programming (see~\cite{Matousek11}). 
\begin{lemma}\label{lm:speclb}
  For any $m \times n$ matrix $A$, and any $m\times m$ diagonal matrix
  $P \geq 0$ with $\tr(P) = 1$, we have
    $\vecdisc(A) \geq \sqrt{n} \sigma_{\min}(P^{1/2}A)$.
\end{lemma}

We define a \emph{spectral lower bound} based on
Lemma~\ref{lm:speclb}.
\begin{equation*}
  \specLB(A) \triangleq \max_{k = 1}^n \max_{S \subseteq [n]: |S| = k}
  \max_P \sqrt{k} \sigma_{\min}(P^{1/2}A|_S),
\end{equation*}
where $P$ ranges over positive (i.e.~$P\succeq 0$) $m \times m$
diagonal matrices satisfying $\tr(P) = 1$. Lemma~\ref{lm:speclb}
implies immediately that $\hvecdisc(A) \geq \specLB(A)$.

Notice that it is not clear whether $\specLB(A)$ can be computed
efficiently. One of our main contributions is to develop a lower bound
on hereditary (vector) discrepancy which is tractable and can be
related to $\specLB(A)$. 

\junk{\begin{cor}\label{cor:speclb}
  For any matrix $A$,
  $\hvecdisc(A) \geq \specLB(A)$.
\end{cor}

Giving a ``global'' upper bound simultaneously for the vector
discrepancy of all submatrices of $A$ is more challenging, and is one
of the main technical challenges of our work.}

\subsection{Vector Balancing and Banaszczyk's Theorem}

A well-known conjecture by Koml\'{o}s states that $\disc(A) \leq C
\|A\|_{1\rightarrow 2}$, for an absolute constant $C$. The
conjecture remains open, and the best known result towards resolving
it is a discrepancy bound of $O(\sqrt{\log n}\ \|A\|_{1\rightarrow 2})$ due to
Banaszczyk~\cite{bana}. Banaszczyk's result in fact concerns the more
general vector balancing problem of determining sufficient conditions
under which there exists an assignment of signs $x \in \{\pm 1\}^n$
  such that $Ax \in K$ for given $m \times n$ matrix $A$ and convex
  body $K$. This general version of Banaszczyk's result is crucial to
our argument.

\begin{theorem}[\cite{bana}]\label{thm:bana}
  There exists a universal constant $C$ such that the following
  holds. Let $A$ be an $m$ by $n$ real matrix such
  that $\|A\|_{1 \rightarrow 2}$, and let $K$ be a convex
  body in $\R^m$ such that $\Pr[g \in K] \geq 1/2$ where $g \in \R^m$
  is a standard $m$-dimensional Gaussian random vector, and the
  probability is taken over the choice of $g$. Then there exists $x
  \in \{-1, 1\}^n$ such that $Ax \in CK$.
\end{theorem}

\junk{Interestingly, while this is a universal bound for matrices
satisfying a specific condition, we use it to give a relative
guarantee, i.e. to approximate the discrepancy of any matrix. }

Another partial result towards the Koml\'{o}s conjecture is a recent bound
by Nikolov~\cite{komlos-sdp} on the vector discrepancy of matrices $A$
satisfying the condition $\|A\|_{1 \rightarrow 2} \leq 1$. Here we
state a version of the bound that is stronger than the one stated
in~\cite{komlos-sdp}. However, a minor variation of the same
proof shows this stronger bound; we give the argument in
Appendix~\ref{app:vector-komlos}. 

\begin{theorem}[\cite{komlos-sdp}]\label{thm:komlos-sdp}
  For any $m \times n$ matrix $A$ satisfying
  $\|A\|_{1 \rightarrow 2} \leq 1$, there exists a $n\times n$
  matrix $X \succeq 0$ such that $\forall i \in [n]: X_{jj} = 1$ and
  $\|AXA^T\|_2 \leq 1$. 
\end{theorem}

Note that the above theorem implies $\vecdisc(A) \leq 1$ because
$\|AXA^T\|_2 \leq 1$ implies $e_i^T(AXA^T)e_i \leq 1$ for all standard
basis vectors $e_i$. However, the spectral norm bound is formally
stronger then the vector discrepancy upper bound, and this will be
essential in our proofs.

\section{Ellipsoid Upper Bounds on Discrepancy}
\label{sect:ub}

In this section we show that small-width ellipsoids provide upper
bounds on both hereditary vector discrepancy and hereditary
discrepancy. Giving such an upper bound is in general challenging
because it must hold for all submatrices simultaneously. The proofs
use Theorems~\ref{thm:bana} and~\ref{thm:komlos-sdp}. We start with
the two main technical lemmas.

\begin{lemma}\label{lm:ellips-bound-vector}
  Let $A = (a_j)_{j=1}^n \in \R^{m \times n}$, and let $F\in
  \R^{m\times m}$ be a rank $m$ matrix such that $\forall j \in [n]:
  a_j \in E = FB_2^m$. Then there exists a matrix $X \succeq 0$ such
  that $\forall j \in [n]: X_{jj} = 1$ and $AXA^T \preceq FF^T$.
\end{lemma}
\begin{proof}
  Observe that, $a_j \in E \Leftrightarrow F^{-1}a_j \in B_2^m$. This
  implies $\|F^{-1}A\|_{1 \rightarrow 2} \leq 1$, and, by
  Theorem~\ref{thm:komlos-sdp}, there exists an $X$ with $X_{jj}=1$
  for all $j$ such that $(F^{-1}A)X(F^{-1}A)^T \preceq I$. Multiplying on
  the left by $F$ and on the right by $F^T$, we have $AXA^T \preceq
  FF^T$, and this completes the proof.
\end{proof}

Lemma~\ref{lm:ellips-bound-vector} is our main tool for approximating
hereditary vector discrepancy. By the relationship between vector
discrepancy and discrepancy established by Bansal
(Corollary~\ref{cor:bansal}), this is sufficient for a
poly-logarithmic approximation to hereditary discrepancy. However, to
get tight upper bounds on discrepancy from small width ellipsoids (and
improved approximation ratio), we give a direct argument using
Banaszczyk's theorem.

\begin{lemma}\label{lm:ellips-bound-bana}
  Let $A = (a_j)_{j=1}^n \in \R^{m \times n}$, and let $F\in
  \R^{m\times m}$ be a rank $m$ matrix such that $\forall j \in [n]: a_j
  \in E = FB_2^m$. Then, for any set of vectors $v_1, \ldots, v_k \in
  \R^m$, there exists $x \in \{\pm 1\}^n$ such that $\forall i \in
  [k]: |\langle Ax, v_i\rangle| \leq C\sqrt{(v_i^TFF^Tv_i)\log k}$ for
  a universal constant $C$.
\end{lemma}
\begin{proof}
  Let $P = \{y: |\langle y, v_i\rangle| \leq \sqrt{v_i^TFF^Tv_i}\ \forall i
  \in [k]\}$.  We need to prove that there exists an $x \in \{-1,
  1\}^n$ such that $Ax \in (C\sqrt{\log k})P$ for a suitable constant
  $C$. Notice that the polar body of $P$ is
  \[
  P^\circ =  \conv\{\pm (v_i^TFF^Tv_i)^{-1/2}v_i\}_{i = 1}^k.
  \] 
  Set $K = F^{-1}P$. To show that there exists an $x$ such that $Ax \in
  (C\sqrt{\log k})P$, we will show that there exists an $x \in \{-1,
  1\}^n$ such that $F^{-1}Ax \in (C\sqrt{\log k})K$. For this, we will
  use Theorem~\ref{thm:bana}. As in the proof of
  Lemma~\ref{lm:ellips-bound-vector}, $\|F^{-1}A\|_{1\rightarrow 2}
  \leq 1$. To use Theorem~\ref{thm:bana}, we also need to argue that
  for a standard Gaussian $g$, $\Pr[g \in (C\sqrt{\log k})K] \geq
  \frac{1}{2}$. To this end, we compute the polar body of $K$ as
  \begin{align*}
  K^\circ
  &= \{y: \langle y, F^{-1}x \rangle \leq 1\; \forall x \in P\}
  = \{y: \langle (F^T)^{-1}y, x \rangle \leq 1\; \forall x \in P\}\\
  &= \{F^Tz: \langle z, x \rangle \leq 1\; \forall x \in P\}
  = F^TP^\circ
  \end{align*}
  By the definition of Minkowski norm, for any $t \in \R$, $y \in tK$,
  if and only if
  \begin{equation*}
    t \geq \|y\|_K = \sup_{z\in K^\circ}\langle y, z \rangle 
    = \sup_{z\in P^\circ}\langle y, F^Tz \rangle =
    \max_{i=1}^k{\frac{1}{\sqrt{v_i^TFF^Tv_i}} |\langle y, F^Tv_i \rangle|},
  \end{equation*}
  where the first equality is by the duality of $\|\cdot\|_K$ and
  $\|\cdot\|_{K^\circ}$, and the final equality holds because the linear functional
  $\langle y, F^Tz \rangle$ is maximized at a vertex of $P^\circ$. We
  have then that $y \in tK$ if and only if
  $\forall i \in [k]: |\langle y, F^Tv_i\rangle|^2 \leq
  t^2(v_i^TFF^Tv_i)$. Let $g$ be a standard $m$-dimensional Gaussian
  vector. Then $\E_g |\langle g, F^Tv_i\rangle|^2 = v_iFF^Tv_i$; by
  standard concentration bounds, $\Pr[|\langle g, F^Tv_i\rangle|^2 >
  t^2 (v_iFF^Tv_i)] < \exp(-t^2/2)$. Setting $t = \sqrt{2\ln 2k}$ and
  taking a union bound over all $i \in [k]$ gives us that $\Pr[g \not
  \in \sqrt{2\ln 2k}\ K] < 1/2$. By Theorem~\ref{thm:bana}, this
  implies that there exists an $x\in \{-1,1\}^n$ such that $F^{-1}Ax \in \sqrt{2\ln
    2k}\ K$, and, by multiplying on both sides by $F$, it follows that
  $Ax \in \sqrt{2\ln 2k}\ P$. 
\end{proof}

Notice that the quantity $\sqrt{v_i^TFF^Tv_i}$ is just the width of
$E$ in the direction of $v_i$. The property that all columns of a
matrix $A$ are contained in $E$ is \emph{hereditary}: if it is
satisfied for $A$, then it is satisfied for any submatrix of $A$. This
elementary fact lends the power of Lemmas~\ref{lm:ellips-bound-vector}
and~\ref{lm:ellips-bound-bana}: the bound given by ellipsoids is
\emph{universal} in the sense that the discrepancy bound for any
direction $v_i$ holds for all submatrices $A|_S$ of $A$
simultaneously.  This fact makes it possible to upper bound hereditary
discrepancy in arbitrary norms, and in the sequel we do this for
$\ell_\infty^m$, which is the norm of interest for standard
definitions of discrepancy. We consider ellipsoids $E$ that contain
the columns of $A$ and minimize the quantity $\|E\|_\infty$: the
largest $\ell_\infty$ norm of the points of $E$. Note that
$\|E\|_\infty$, for an ellipsoid $E = FB_2^m$, can be written as
\begin{equation}\label{eq:ellips-infty}
  \|E\|_\infty = \max_{x \in E, y: \|y\|_1 = 1}{\langle x, y \rangle}
  = \max_{y:\|y\|_1 = 1}{\|y\|_{E^\circ}} = \max_{i \in [n]}{\sqrt{e_i^T
    FF^T e_i}}, 
\end{equation}
where the first identity follows since $\ell_1$ is the dual norm to
$\ell_\infty$, and the final identity follows from the formula for
$\|\cdot\|_{E^\circ}$ and the fact that a convex function over the
$\ell_1$ ball is always maximized at a vertex, i.e. a standard basis
vector. The next theorem gives our main upper bound on hereditary
(vector) discrepancy, which is in terms of $\|E\|_\infty$.

\begin{theorem}\label{thm:main-ub}
  Let $A = (a_i)_{i=1}^n \in \R^{m \times n}$, and let $F$ be a rank
  $m$ matrix such that $\forall i \in [n]: a_i \in E = FB_2^m$. Then
  $\hvecdisc(A) \leq \|E\|_\infty$, and $\herdisc(A) = O(\sqrt{\log
    m}) \|E\|_\infty$.
\end{theorem}
\begin{proof}
  Let $A|_S$ be an arbitrary submatrix of $A$ ($S \subseteq
  [n]$). Since all columns of $A$ are contained in $E$, this holds for
  all columns of $A|_S$ as well, and by Lemma~\ref{lm:ellips-bound-vector},
  we have that there exists $X \succeq 0$ with $X_{jj} = 1$ for all $j
  \in S$, and $(A|_S)X(A|_S)^T \preceq FF^T$. Therefore, for all $i
  \in [m]$, $e_i^T(A|_S)X(A|_S)^Te_i \leq e_iFF^Te_i \leq
  \|E\|_\infty^2$, by \eqref{eq:ellips-infty}. Since $S$ was
  arbitrary, this implies the bound on $\hvecdisc(A)$. For bounding
  $\herdisc(A)$, in
  Lemma~\ref{lm:ellips-bound-bana}, set $k=m$ and $v_i = e_i$ for $i \in
  [m]$ and $e_i$ the $i$-th standard basis vector.
\end{proof}

\section{Lower Bounds on Discrepancy}

In Section~\ref{sect:ub} we showed that the hereditary (vector)
discrepancy of a matrix $A$ can be \emph{upper bounded} in terms of the $\|E
\|_\infty$ for any $E$ containing the columns of $A$. In this section
we analyze the properties of the minimal such ellipsoid and show that it provides
lower bounds for discrepancy as well. We use convex duality and the
restricted invertibility theorem for this purpose. The lower bound we
derive is new in discrepancy theory and of independent interest.

\subsection{The Ellipsoid Minimization Problem and Its Dual}
\label{sect:min-ellips}

\junk{In the remainder of this section we fix an $m \times n$ matrix $A =
(a_j)_{j = 1}^n$. To avoid cumbersome technicalities, we assume that
$\rank(A) = m$; this comes without loss of generality, because it can
be guaranteed by applying an arbitrary small perturbation to $A$ and
adding sufficiently many columns very close to $0$. Both these
transformations can be made small enough so that they perturb the
discrepancy of $A$ and the value of the optimization problems we
consider only negligibly.}

\confoption{}{To formulate the problem of minimizing $\|E\|_\infty = \max_{x \in
  E}{\|x\|_\infty}$ as a convex optimization problem we need the
following well-known lemma, which shows that the matrix inverse is
convex in the PSD sense.
\begin{lemma}\label{lm:inverse-convex}
  For any two $m\times m$ matrices $X \succ 0$ and $Y \succ 0$,
  $(\frac{1}{2}X + \frac{1}{2}Y)^{-1} \preceq \frac{1}{2}X^{-1} +
  \frac{1}{2}Y^{-1}$. 
\end{lemma}
\begin{proof}
  Define the matrices
  \begin{equation*}
    U = \left(
      \begin{array}{cc}
        X^{-1} &I\\
        I &X
      \end{array}
      \right)
      \;\;\;\;\;
      V = \left(
      \begin{array}{cc}
        Y^{-1} &I\\
        I &Y
      \end{array}
      \right).
  \end{equation*}
  The Schur complement of $X$ in $U$ is $0$, and
  therefore $U \succeq 0$, and analogously $V \succeq 0$. Therefore
  $U + V \succeq 0$, and the Schur complement of $X + Y$ in $U+V$ is also
  positive semidefinite, i.e.~$X^{-1} + Y^{-1} - 4(X+Y)^{-1} \succeq
  0$. This completes the proof, after re-arranging terms.
\end{proof}}

Consider a matrix $A = (a_j)_{j=1}^n \in \R^{m\times n}$ of rank
$m$. Let us formulate $\min \{\|E\|_\infty: \forall j \in [n]: a_j
\in E\}$ as a convex minimization problem. The problem is defined as
follows
\begin{align}
  &\text{Minimize } t \label{eq:ellips-obj}
  \text{ s.t. }\\
  &X\succ 0\\
  &\forall i \in [m]: e_i^TX^{-1}e_i \leq t\label{eq:ellips-width}\\
  &\forall j \in [n]: a_j^TXa_j \leq 1.\label{eq:ellips-enclose}
\end{align}

\begin{lemma}\label{lm:ellips-program}
  For a rank $m$ matrix $A = (a_j)_{j = 1}^n \in \R^{m\times n}$, the
  optimal value of the optimization problem
  \eqref{eq:ellips-obj}--\eqref{eq:ellips-enclose} is equal to $\min
  \{\|E\|_\infty^2: \forall j \in [n]: a_j \in E\}$. Moreover, the
  objective function \eqref{eq:ellips-obj} and constraints
  \eqref{eq:ellips-width}--\eqref{eq:ellips-enclose} are convex
  over $t \in \R$ and $X \succ 0$.
\end{lemma}
\confoption{\begin{proof}[Proof Sketch]The equality of optimal values
    follows from the basic properties of ellipsoids
    and~\eqref{eq:ellips-infty}. The convexity of
    \eqref{eq:ellips-obj}--\eqref{eq:ellips-enclose} is a consequence
    of the convexity of the matrix inverse, which is well-known, see
    e.g.~\cite[Chapter 1.5]{bhatia-psdbook}. We give the full proof
    in Appendix~\ref{app:dual}.\end{proof}}{\begin{proof}
    Let $\lambda$ be the optimal value of
    \eqref{eq:ellips-obj}--\eqref{eq:ellips-enclose} and $\mu =
    \min\{\|E\|_\infty: \forall j \in [n]: a_j \in E\}$. Given a
    feasible $X$ for \eqref{eq:ellips-obj}--\eqref{eq:ellips-enclose},
    set $E = X^{-1/2}B_2^m$ (this is well-defined since $X \succ
    0$). Then for any $j \in [n]$, $\|a_j\|_E = a_j^TXa_j \leq 1$ by
    \eqref{eq:ellips-enclose}, and, therefore, $a_j \in E$. Also, by
    \eqref{eq:ellips-infty}, $\|E\|_\infty^2 = \max_{i = 1}^m
    e_i^TXe_i \leq t$. This shows that $\mu \leq \lambda$. In the
    reverse direction, let $E = FB_2^m$ be such that $\forall j\in
    [n]: a_j \in E$. Then, because $A$ is full rank, $F$ is also full
    rank and invertible, and we can define $X = (FF^T)^{-1}$ and $t =
    \|E\|_\infty^2$. Analogously to the calculations above, we can
    show that $X$ and $t$ are feasible, and therefore $\lambda \leq
    \mu$.

  The objective function and the constraints \eqref{eq:ellips-enclose}
  are affine, and therefore convex. To show \eqref{eq:ellips-width} are
  also convex, let $X_1, t_1$ and $X_2, t_2$ be two feasible solutions. Then,
  Lemma~\ref{lm:inverse-convex} implies that 
  for any $i$, $e_i^T(\frac{1}{2}X_1 + \frac{1}{2}X_2)^{-1}e_i \leq
  \frac{1}{2}X_1^{-1} + \frac{1}{2}X_2^{-1} \leq \frac{1}{2}t_1 +
  \frac{1}{2}t_2$, and constraints \eqref{eq:ellips-width} are convex as well.
\end{proof}}

\begin{theorem}\label{thm:nuclear}
  Let $A = (a_j)_{j = 1}^n \in \R^{m\times n}$ be a rank $m$
  matrix, and let $\mu = \min
  \{\|E\|_\infty: \forall j \in [n]: a_j \in E\}$. Then,
  \begin{align}
    \mu^2 =   &\max \|P^{1/2}AQ^{1/2}\|_{S_1}^2\label{eq:nuclear-obj2}
    \text{ s.t.}\\
    &\tr(P) = \tr(Q) = 1\\
    &P,Q \succeq 0; P,Q \text{ diagonal}. \label{eq:nuclear-pos2}  
  \end{align}
\end{theorem}
\begin{proof}
  \confoption{We prove the theorem by showing that the convex
    optimization problem
    \eqref{eq:ellips-obj}--\eqref{eq:ellips-enclose} satisfies
    Slater's condition (so strong duality holds), and its Lagrange
    dual is equivalent to
    \eqref{eq:nuclear-obj2}--\eqref{eq:nuclear-pos2}. Strong duality
    then implies that the optimal value of
    \eqref{eq:ellips-obj}--\eqref{eq:ellips-enclose} is equal to the
    optimal value of the Lagrange dual, which is equal to the optimal
    value of \eqref{eq:nuclear-obj2}--\eqref{eq:nuclear-pos2}. The
    Lagrange dual problem of
    \eqref{eq:ellips-obj}--\eqref{eq:ellips-enclose} is to maximize
    $2\|P^{1/2}AR^{1/2}\|_{S_1} - \tr(R)$ over diagonal PSD matrices
    $P, R$, such that $\tr(P) = 1$. (Note that, as a Lagrange dual,
    this is a convex optimization problem.) The objective of the
    Lagrange dual is maximized when $\tr(R) =
    \|P^{1/2}AR^{1/2}\|_{S_1}^2$, and, therefore, the optimal value of
    the dual is equal to the optimal value of the optimization
    problem~\eqref{eq:nuclear-obj2}--\eqref{eq:nuclear-pos2}. The full
    proof is given in Appendix~\ref{app:dual}.}{We shall prove the
    theorem by showing that the convex optimization problem
    \eqref{eq:ellips-obj}--\eqref{eq:ellips-enclose} satisfies
    Slater's condition, and its Lagrange dual is equivalent to
    \eqref{eq:nuclear-obj2}--\eqref{eq:nuclear-pos2}. Let us first
    verify Slater's condition. We define the domain for constraints
    \eqref{eq:ellips-width} as the open cone $\{X: X \succ 0\}$, which
    makes the constraint $X \succ 0$ implicit. Let $r =
    \|A\|_{1\rightarrow 2}$, $X = \frac{1}{r}I$, and $t =
    r+\varepsilon$ for some $\varepsilon > 0$. Then the affine
    constraints \eqref{eq:ellips-enclose} are satisfied exactly, and
    the constraints \eqref{eq:ellips-width} are satisfied with slack
    since $\varepsilon > 0$. Moreover, by
    Lemma~\ref{lm:ellips-program}, all the constraints and the
    objective function are convex. Therefore,
    \eqref{eq:ellips-obj}--\eqref{eq:ellips-enclose} satisfies
    Slater's condition, and consequently strong duality holds.

  The Lagrange dual function for
  \eqref{eq:ellips-obj}--\eqref{eq:ellips-enclose} is
  \begin{equation*}
    g(p,q) = \inf_{t, X \succ 0}{t + \sum_{i = 1}^m{p_i(e_i^TX^{-1}e_i
        - t)} + \sum_{j = 1}^n{q_j(a_j^TXa_j - 1)} },
  \end{equation*}
  with dual variables $p \in \R^m$ and $q \in \R^n$, $p, q \geq
  0$. Equivalently, writing $p$ as a diagonal matrix $P \in
  \R^{m\times m}$, $P \succeq 0$, $q$ as a diagonal matrix $R \in
  \R^{n\times n}$, $R \succeq 0$, we have $g(P,R) = \inf_{t,X\succ
    0}{t + \tr(PX^{-1}) - \tr(tP) + \tr(ARA^TX) - \tr(R)}$.  If $\tr(P)
    \neq 1$, then $g(P,R) = -\infty$, since we can take $t$ to
    $-\infty$ while keeping $X$ fixed. On the other hand, for $\tr(P)
    = 1$, the dual function simplifies to
  \begin{equation}\label{eq:g-raw}
    g(P,R) = \inf_{X \succ 0}{\tr(PX^{-1}) + \tr(ARA^TX) - \tr(R)}.
  \end{equation}
  Since $X\succ 0$ implies $X^{-1}\succ 0$, $g(P,R) \geq -\tr(R) >
  -\infty$ whenever $\tr(P) = 1$. Therefore, $g(P,R)$ is continuous
  over the set of diagonal positive semidefinite $P$, $R$ such that
  $\tr(P) = 1$. For the rest of the proof we assume that $P$ and
  $ARA^T$ are rank $m$. This is without loss of generality by the
  continuity of $g$ and because both assumptions can be satisfied by
  adding arbitrarily small perturbations to $P$ and $R$. (Here we use
  the fact that $A$ is rank $m$.)
  
  After differentiating the right hand side of \eqref{eq:g-raw} with
  respect to $X$, we get the first-order optimality condition 
  \begin{equation}
    \label{eq:fo-optimality}
    X^{-1}PX^{-1} =ARA^T.
  \end{equation}
  Multiplying by $P^{1/2}$ on the left and the right and taking square
  roots gives the equivalent condition $P^{1/2}X^{-1}P^{1/2} =
  (P^{1/2}ARA^TP^{1/2})^{1/2}$. This equation has a unique solution,
  since $P$ and $ARA^T$ were both assumed to be invertible. Since
  $\tr(PX^{-1}) = \tr(P^{1/2}X^{-1}P^{1/2})$ and also, by
  \eqref{eq:fo-optimality}, $\tr(ARA^TX) = \tr(X^{-1}P) =
  \tr(PX^{-1})$, we simplify $g(P,R)$ to
  \begin{equation}\label{eq:g-final}
    g(P,R) = 2\tr((P^{1/2}ARA^TP^{1/2})^{1/2}) - \tr(R) =
    2\|P^{1/2}AR^{1/2}\|_{S_1} - \tr(R).
  \end{equation}
  We showed that
  \eqref{eq:ellips-obj}--\eqref{eq:ellips-enclose} satisfies Slater's
  condition and therefore strong duality holds, so by Theorem~\ref{thm:slater} and
  Lemma~\ref{lm:ellips-program}, 
  $\mu^2 = \max\{g(P,R): \tr(P) = 1, P,R \succeq 0, \text{ diagonal}\}$.
  Let us define new variables $Q$ and $c$, where $c = \tr(R)$ and $Q =
  R/c$. Then we can re-write $g(P,R)$ as 
  \begin{equation*}
    g(P,R) = g(P,Q,c) = 2\|P^{1/2}A(cQ)^{1/2}\|_{S_1} - \tr(cQ) =
    2\sqrt{c}\|P^{1/2}AQ^{1/2}\|_{S_1} - c. 
  \end{equation*}
  From the first-order optimality condition $\frac{dg}{dc} =
  0$, we see that maximum of $g(P,Q,c)$ is achieved when $c =
  \|P^{1/2}AQ^{1/2}\|_{S_1}^2$ and is equal to
  $\|P^{1/2}AQ^{1/2}\|_{S_1}^2$. Therefore, maximizing $g(P,R)$ over
  diagonal positive semidefinite $P$ and $R$ such that $\tr(P) = 1$ is
  equivalent to the optimization problem
  \eqref{eq:nuclear-obj2}--\eqref{eq:nuclear-pos2}. This completes the proof.}
\end{proof}

\subsection{Spectral Lower Bounds via Restricted Invertibility}

In this subsection we relate the dual formulations of the
min-ellipsoid problem from Section~\ref{sect:min-ellips} to the dual
of vector discrepancy, and $\specLB$ in particular. The connection is
via the restricted invertibility principle and gives our main lower
bounds on hereditary (vector) discrepancy.

\begin{lemma}\label{lm:bt-lb}
  Let $A$ be an $m$ by $n$ real matrix, and let $Q \succeq 0$ be a diagonal
  matrix such that $\tr(Q) = 1$. Then there exists a submatrix $A|_S$
  of $A$ such that $|S|\sigma_{\min}(A|_S)^2 \geq
  \frac{c^2\|AQ^{1/2}\|_{S_1}^2}{(\log m)^2}$, for a universal
  constant $c > 0$. Moreover, given $A$ as input, $S$ can be
  computed in deterministic polynomial time. 
\end{lemma}
\begin{proof}
  By homogeneity of the nuclear norm and the smallest singular value,
  it suffices to show that if $\|AQ^{1/2}\|_{S_1}^2 = 1$, then
  $|S|\sigma_{\min}(A|_S)^2 \geq \frac{c^2}{(\log m)^2}$ for a set $S
  \subseteq [n]$. 

  Let $T_k = \{i \in [m]: 2^{-k-1}\leq \sigma_i(AQ^{1/2}) \leq
  2^{-k}\}$ for an integer $0 \leq k\leq \log_2 m$, and $R = \{i \in
  [m]: \sigma_i(AQ^{1/2}) \leq \frac{1}{2m}\}$. Then 
  $$
  \sum_{k = 0}^{\log_2 m}\sum_{i \in T_k}{\sigma_i(AQ^{1/2})} 
  = 1 - \sum_{i \in R}{\sigma_i(AQ^{1/2})}\geq  1/2,
  $$
  since $|R| \leq m$. Therefore, by averaging, there exists a $k^*$
  such that $\sum_{i \in T_{k^*}}{\sigma_i(AQ^{1/2})} \geq
  \frac{1}{2\log_2 m}$. Let $\Pi$ be the projection operator onto the
  span of the left singular vectors of $AQ^{1/2}$ corresponding to the
  singular values $\sigma_i(AQ^{1/2})$ for $i \in T_{k^*}$. Setting
  $\tau = \frac{1}{2\log_2 m}$ and $r = |T_{k^*}| = \rank(\Pi
  AQ^{1/2})$, we have $\|\Pi AQ^{1/2}\|_{S_1} \geq \tau$ by the choice
  of $k^*$, and $\|\Pi AQ^{1/2}\|_2 \leq 2\tau/r$ because all values
  of $\Pi AQ^{1/2}$ are within a factor of $2$ from each
  other. Finally, applying Cauchy-Schwarz to the singular values of
  $\Pi AQ^{1/2}$, we have that $\|\Pi AQ^{1/2}\|_{HS} \geq
  \tau/r^{1/2}$. By Theorem~\ref{thm:bt} applied with $\epsilon =
  \frac{1}{2}$, there exists a set $S$ of size $|S| \geq r/16$ such
  that $\sigma_{\min}(\Pi A|_S)^2 \geq \tau^2/4r$, implying that
  \begin{equation*}
    |S|\sigma_{\min}(A|_S)^2 \geq |S|\sigma_{\min}(\Pi A|_S)^2 \geq \frac{1}{64}\tau^2.
  \end{equation*}
  Moreover, $S$ can be computed in deterministic polynomial time.
\end{proof}

\begin{theorem}\label{thm:main-lb}
  Let $\mu= \min \{\|E\|_\infty: \forall j \in [n]: a_j \in E\}$ for a rank $m$
  matrix $A = (a_j)_{j = 1}^n$. Then
  $$
  \mu = O(\log m)\ \hvecdisc(A).
  $$
  Moreover, we can compute in deterministic polynomial time a set $S
  \subseteq [n]$ such that $\mu =  O(\log
  m)\vecdisc(A|_S)$. 
\end{theorem}
\begin{proof}
  Let $P$ and $Q$ be optimal solutions for
  \eqref{eq:nuclear-obj2}-\eqref{eq:nuclear-pos2}. By
  Theorem~\ref{thm:nuclear}, $\mu =
  \|P^{1/2}AQ^{1/2}\|_{S_1}$. Then, by Lemma~\ref{lm:bt-lb}, applied
  to the matrices $P^{1/2}A$ and $Q$, there exists a set $S \subseteq
  [n]$, computable in deterministic polynomial time, such that 
  \begin{equation}\label{eq:main-lb}
    \specLB(A) \geq \sqrt{|S|}\sigma_{\min}(P^{1/2}A|_S) \geq \frac{c
      \|P^{1/2}AQ^{1/2}\|_{S_1}}{\log m} = \frac{c\mu}{\log m}.
  \end{equation}
\end{proof}

The determinant lower bound of Lovasz, Spencer, and
Vesztergombi~\cite{LSV} is equal to the maximum of $|\det A_{S,T}|$
over all submatrices $A_{S,T}$ of $A$. We note that up to the log
factor, the lower bound~\eqref{eq:main-lb} is at least as
strong. In particular, assume the determinant lower bound is maximized by
a $k\times k$ submatrix $A_{S,T}$ induced by a subset $S$ of the rows
and a subset $T$ of the columns. Then we can set $P =
\frac{1}{k}\diag(\mathbf{1}_S)$  and $Q =
\frac{1}{k}\diag(\mathbf{1}_T)$, and $\|P^{1/2}AQ^{1/2}\|_{S_1} =
\frac{1}{k} \|A_{S,T}\|_{S_1}$ is at least as large as $|\det
A_{S,T}|$ by the geometric mean-arithmetic mean inequality applied to
the singular valus of $A_{S,T}$.

\section{The Approximation Algorithm}

We are now ready to give our approximation algorithm for hereditary
vector discrepancy and hereditary discrepancy. In fact, the algorithm
is a straightforward consequence of the upper and lower bounds we
proved in the prior sections.

\begin{theorem}[Theorems~\ref{thm:vecdisc-apx},\ref{thm:herdisc-apx},
  restated] 
  There exists a polynomial time algorithm that, on input an $m \times
  n$ real matrix $A = (a_j)_{j=1}^n$, computes a value $\alpha$ such
  that the following inequalities hold
  \begin{align*}
    \alpha &\leq \hvecdisc(A) \leq O(\log m) \cdot \alpha\\
    \alpha &\leq \herdisc(A) \leq O(\log^{3/2} m) \cdot \alpha
  \end{align*}
  Moreover, the algorithm finds a submatrix $A|_S$ of $A$,
  such that $\alpha \leq \vecdisc(A|_S)$.
\end{theorem}
\begin{proof}
  We first ensure that the matrix $A$ is of rank $m$ by adding a tiny
  full rank perturbation to it, and adding extra columns if
  necessary\footnote{There are other, more numerically stable ways to
    reduce to the full rank case, e.g.~by projecting $A$ onto its
    range and modifying the norms we consider accordingly. We choose
    the perturbation approach for simplicity.}. By making the
  perturbation small enough, we can ensure that it affects
  $\herdisc(A)$ and $\hvecdisc(A)$ negligibly. Let $\mu =
  \min\{\|E\|_\infty: \forall j \in [n]: a_j \in E\}$.  The value
  $\alpha$ we output is $\alpha = \mu/(C\log m)$, where $C$
  is a sufficiently large constant so that the asymptotic expression
  in Theorem~\ref{thm:main-lb} holds. The approximation guarantees
  follow from Theorems~\ref{thm:main-ub} and~\ref{thm:main-lb}, and
  $S$ is computed as in Theorem~\ref{thm:main-lb}.

  To compute $\alpha$ in polynomial time, we solve
  \eqref{eq:ellips-obj}--\eqref{eq:ellips-enclose}. By
  Lemma~\ref{lm:ellips-program}, this is a convex minimization
  problem, and as such can be solved using the ellipsoid method up to
  an $\epsilon$-approximation in time polynomial in the input size and
  in $\log \epsilon^{-1}$. The optimal value is equal to
  $\mu$ by Lemma~\ref{lm:ellips-program}, and, therefore,
  we can compute an arbitrarily good approximation to $\alpha$ in
  polynomial time.  
\end{proof}

{
\section{A Geometric Consequence}

In this section we derive a geometric consequence of
Lemma~\ref{lm:bt-lb}. Specifically, we prove that any convex body $K$ is
contained in an ellipsoid whose Gaussian width is bounded in terms of
the Kolmogorov widths of $K$. While not necessary for our
approximation algorithm, this result may be of independent interest.

Let us first define the Kolmogorov widths for a convex body $K$.
\begin{definition}
  The Kolmogorov width $d_k(K)$ of a symmetric convex body $K\subseteq
  \R^n$ is equal to $d_k(K) \triangleq \min_{\Pi \in
    \mathcal{P}_{n-k+1}}{\|\Pi K\|_2}$, i.e. the minimum radius (in $\ell_2$) of any
  projection of $K$ of co-dimension $k-1$.
\end{definition}
We note that Kolmogorov width is more generally defined for linear
operators between Banach spaces, and the definition above is the
special case of the Kolmogorov width of the identity operator $I:X
\rightarrow \ell_2$, where $X$ is a finite-dimensional Banach space
with unit ball $K$.

Lemma~\ref{lm:bt-lb} implies the following result.
\begin{theorem}\label{thm:geometric}
  Let $K \subseteq \R^n$ be a
  symmetric convex body. There exists an ellipsoid $E=FB_2^n$
  containing $K$ such that
  \begin{equation*}
    \|F\|_{HS} \leq (C\log n)\max_{k = 1}^n{\sqrt{k} d_k(K)} ,
  \end{equation*}
  for a universal constant $C$.
\end{theorem}

The proof of the result relies on Lemma~\ref{lm:bt-lb} and an
optimization problem over ellipsoids containing $K$ that is closely
related to the problem of minimizing width in coordinate directions,
discussed in prior sections. Let $v_1, \ldots, v_N$ be points in
$\R^n$, and consider the problem of minizing $\|F\|_{HS}^2$ subject to
$v_1, \ldots, v_N \in E = FB_2^n$. The problem can be formulated as
\begin{align}
  &\text{Minimize } \tr(X^{-1}) \label{eq:trace-obj}
  \text{ s.t. }\\
  &X\succ 0\\
  &\forall i \in [N]: v_i^TXv_i \leq 1.\label{eq:enclose-tr}
\end{align}
The constraints \eqref{eq:enclose-tr} are affine, and the convexity of
the objective \eqref{eq:trace-obj} follows from
Lemma~\ref{lm:inverse-convex}. An argument analogous to the one in the
proof of Theorem~\ref{thm:nuclear} shows that the Lagrange dual
function for \eqref{eq:trace-obj}--\eqref{eq:enclose-tr} is
\begin{equation*}
  g(R) = \|VR^{1/2}\|_{S_1} - \tr(R),
\end{equation*}
where $V$ is the matrix whose columns are $v_1, \ldots, v_N$, and $R$
is a non-negative $N\times N$ diagonal matrix. Again analogously to
Theorem~\ref{thm:nuclear}, the Lagrange dual problem is to maximize
$g(R)$ over all non-negative diagonal $R$, and by strong duality we have the equality
\begin{align}
\min\{\|F\|_{HS}^2: \forall i \in [N]\ v_i \in E = FB_2^n\} &= \max\{g(R): R \succeq 0,
  \text{diagonal}\} \label{eq:mintrace-dual}\\
  &= \max\{ \|VQ^{1/2}\|^2_{S_1}: Q \succeq 0,
    \text{diagonal}, \tr(Q) = 1\}. \notag
\end{align}

\junk{The optimization problem \eqref{eq:trace-obj}--\eqref{eq:enclose-tr} can
be formulated equivalently as an optimization problem over a compact
set, and therefore the minimum is achieved. We call the minimizing
ellipsoid the \emph{min-trace ellipsoid} for $v_1, \ldots, v_N$.}

\vspace{1em}
\begin{proofof}{Theorem~\ref{thm:geometric}}
  Let $v_1, \ldots, v_N$ be chosen to form a sufficiently dense net on
  the boundary of $K$ and let $E = FB_2^n$ be the ellipsoid containing
  $v_1, \ldots, v_N$ that minimizes $\|F\|_{HS}^2$; by taking $N$
  sufficiently large, we can ensure that $K \subseteq (1+\epsilon)E$
  for an arbitrary small $\epsilon$.

  Let $V$ be the matrix whose columns are the points $v_1, \ldots,
  v_N$. By \eqref{eq:mintrace-dual} and Lemma~\ref{lm:bt-lb} (with $V$ used
  in the role of $A$), there exists a set  $S \subseteq [N]$ and an
  absolute constant $c$ such that
  \begin{equation}
    \label{eq:matr-singval-lb}
      |S|\sigma_{\min}(V|_S)^2 \geq \frac{c}{(\log n)^2} \|F\|_{HS}^2.  
  \end{equation}
We claim that for $s = |S|$, $k = \lceil s/2
  \rceil$, and any $\Pi \in \mathcal{P}_{n-k+1}$ there exists an $i \in S$
  such that $\|\Pi v_i\|_2^2 \geq \sigma_{\min}(V|_S)^2/2$. This
  suffices to prove the theorem, since, together
  with \eqref{eq:matr-singval-lb}, it implies that
  $kd_{k}(K)^2 \geq \frac{c/2}{(\log n)^2} \|F\|_{HS}$.

  Define $M \triangleq (V|_S)(V|_S)^T$ and fix some $\Pi \in
  \mathcal{P}_{n-k+1}$.   By averaging, it suffices to
  show that 
  \[
  \frac{1}{s} \sum_{i \in S}{\|\Pi v_i\|_2^2} 
  = \frac{1}{s}\tr((V|_S)^T\Pi(V|_S)) \geq \sigma_{\min}(V|_S)^2/2. 
  \]
  Let $\Pi = UU^T$, where $U$ is a matrix
  with $n-k + 1$ mutually orthogonal unit columns.  Then, by the
  Cauchy interlace theorem (see Lemma~\ref{lm:interlace} in
  Appendix~\ref{app:vector-komlos}),
  \begin{equation*}
    \lambda_{k}(U^TMU) \geq \lambda_{2k-1}(M) \geq
    \lambda_{s}(M) =\sigma_{\min}(V|_S)^2. 
  \end{equation*}
  Therefore, we have
  \begin{align*}
    \frac{1}{s}\tr((V|_S)^T\Pi (V|_S)) &=
    \frac{1}{s}\tr((U^TV|_S)^T(U^TV|_S))
    =     \frac{1}{s}\tr((U^TV|_S)(U^TV|_S)^T)\\
    &= \frac{1}{s}\tr(U^TMU)
    \geq  \frac{k}{s}\lambda_k(U^TMU) \geq \frac{1}{2}\sigma_{\min}(V|_S)^2.
  \end{align*}
  As remarked above, this completes the proof of the theorem together
  with \eqref{eq:matr-singval-lb}. 
\end{proofof}
\vspace{1em}

The Hilbert-Schmidt norm $\|F\|_{HS}$ has several natural geometric
interpretations in terms of the ellipsoid $E = FB_2^n$. On one hand,
$\|F\|_{HS}^2$ is equal to the sum of squared lengths (in $\ell_2$) of
the major axes of $E$. By an easy calculation, $\|F\|_{HS}$ is also
equivalent up to constants to the norm $\ell^*(K) \triangleq \E
\|g\|_{K^\circ} = \E \max_{x \in K}{|\langle x, g\rangle |}$, where
$g$ is a standard $m$-dimensional Gaussian random variable (see,
e.g.~\cite{talagrand-chaining}). This quantity is also known as the
Gaussian width of $K$. Phrased in these terms,
Theorem~\ref{thm:geometric} shows that for any $n$-dimensional convex
symmetric $K$ there exists an ellipsoid $E$ containing $K$ such that 
\begin{equation}\label{eq:gauss-width}
  \ell^*(E) \leq (C\log n) \max_{k = 1}^n{\sqrt{k} d_k(K)}. 
\end{equation}
A qualitatively weaker bound follows from a theorem of Carl and
Dudley's chaining bound. Carl~\cite{carl-apxnum} showed
that for an absolute constant $C_1$, 
\[
\max_{k=1}^n{\sqrt{k} e_k(K)} \leq C_1 \max_{k =
  1}^n{\sqrt{k}d_k},
\]
where $e_k(K)$ is the $k$-th entropy number of
$K$, i.e.~the least $r$ such that $K$ can be covered by at most
$2^{k-1}$ copies of $rB_2^n$. Dudley's chaining
argument~\cite{dudley-chaining,talagrand-chaining} implies that there
exists a constant $C_2$ such that
$\ell^*(K) \leq (C_2\log n)\max_{k=1}^n{\sqrt{k} e_k(K)}$; combining the two
bounds we have that 
\begin{equation}\label{eq:carl-dudley}
  \ell^*(K) \leq (C_3\log n) \max_{k = 1}^n{\sqrt{k} d_k(K)},
\end{equation}
where $C_3 = C_1C_2$.
This is readily implied by \eqref{eq:gauss-width} (up to the value of
the constant), because $K
\subseteq E$ implies $\ell^*(K) \leq \ell^*(E)$. However, there are
examples where \eqref{eq:gauss-width} is near-tight while
\eqref{eq:carl-dudley} is loose. For example, for the $\ell_1^n$-ball
$B_1^n$, $\max_{k=1}^n \sqrt{k} d_k(B_1^n) = \Omega(\sqrt{n})$ and
$\ell^*(B_1^n) = \Theta(\sqrt{\log n})$, while for any ellipsoid $E$
containing $B_1$ we have $\ell^*(E) =\Omega(\sqrt{n})$, as can be seen
from \eqref{eq:mintrace-dual}.}
\section{Conclusion}

We gave an $O(\log^{3/2} n)$-approximation algorithm for the hereditary discrepancy of a matrix $A$, by approximately characterizing the hereditary vector discrepancy of a matrix in terms of a simple convex program: that of minimizing $\|E\|_\infty$ over all $E$ that contain the columns of $A$.

Our
lower bound is ``constructive'': we can construct in polynomial time a
submatrix $A|_S$ demonstrating the lower bound on $\hvecdisc(A)$ and
hence on $\herdisc(A)$. Our upper bound for $\hvecdisc(A)$ is also
``constructive'' in that the ellipsoid $E^*$ gives a recipe to
construct a vector solution to the $\vecdisc(A|_S)$ given any $S$. Our
$O(\log^{3/2})$ upper bound for $\herdisc(A)$ is however
non-constructive as it uses the result of Banaszczyk~\cite{bana},
whose proof does not yield an efficient algorithm for computing the
sign vector $x$. We
can however use the result of Bansal to algorithmically get a coloring
for any given $S$, at the cost of losing a factor of $\sqrt{\log n}$
in the approximation.

We leave open several questions of interest. One natural question is
whether our approximation ratios can be improved. The best known
hardness of approximating hereditary discrepancy is $2$, but we conjecture
that the hardness is superconstant. Another interesting
question is whether the guarantee for Bansal's algorithm
(Theorem~\ref{thm:bansal}) can be improved by a factor of
$O(\sqrt{\log m})$, which would make it tight. This question was
previously posed by Matou\v{s}ek~\cite{Matousek11}. Such an improvement would
also imply a constructive proof of Banaszczyk's theorem. A further
question concerns the complexity of computing $\herdisc(A)$
exactly. Deciding whether $\herdisc(A) \leq t$ for any $t$ is
naturally in $\Pi^{\mathsf{P}}_2$, but not know to be in
$\mathsf{NP}$. Is this problem complete for $\Pi^{\mathsf{P}}_2$? 


\confoption{}{\section*{Acknowledgements}

The first named author thanks Daniel Dadush for useful discussions,
and Assaf Naor for bringing up the question about a geometric
equivalent of Lemma~\ref{lm:bt-lb}.}

\bibliographystyle{alpha}
\bibliography{discrepancy}

\begin{thebibliography}{AGH13}

\bibitem[AGH13]{AustrinGH13}
Per Austrin, Venkatesan Guruswami, and Johan H{\aa}stad.
\newblock $(2+\epsilon)$-{SAT} is {NP}-hard.
\newblock In {\em ECCC}, 2013.

\bibitem[Ban98]{bana}
Wojciech Banaszczyk.
\newblock Balancing vectors and gaussian measures of n-dimensional convex
  bodies.
\newblock {\em Random Structures \& Algorithms}, 12(4):351--360, 1998.

\bibitem[Ban10]{Bansal10}
N.~Bansal.
\newblock Constructive algorithms for discrepancy minimization.
\newblock In {\em Foundations of Computer Science (FOCS), 2010 51st Annual IEEE
  Symposium on}, pages 3--10. IEEE, 2010.

\bibitem[BS95]{beck-sos}
J\'{o}zsef Beck and Vera~T. S\'{o}s.
\newblock Discrepancy theory.
\newblock In R.~L. Graham, M.~Gr\"{o}tschel, and L.~Lov\'{a}sz, editors, {\em
  Handbook of Combinatorics (vol. 2)}, pages 1405--1446. MIT Press, Cambridge,
  MA, USA, 1995.

\bibitem[BT87]{bour-tza}
J.~Bourgain and L.~Tzafriri.
\newblock Invertibility of large submatrices with applications to the geometry
  of banach spaces and harmonic analysis.
\newblock {\em Israel journal of mathematics}, 57(2):137--224, 1987.

\bibitem[Car81]{carl-apxnum}
Bernd Carl.
\newblock Entropy numbers, $s$-numbers, and eigenvalue problems.
\newblock {\em Journal of Functional Analysis}, 41(3):290--306, 1981.

\bibitem[Cha00]{Chazelle}
Bernard Chazelle.
\newblock {\em The Discrepancy Method: Randomness and Complexity}.
\newblock Cambridge University Press, 2000.

\bibitem[CNN11]{CNN}
M.~Charikar, A.~Newman, and A.~Nikolov.
\newblock Tight hardness results for minimizing discrepancy.
\newblock In {\em Proceedings of the Twenty-Second Annual ACM-SIAM Symposium on
  Discrete Algorithms}, pages 1607--1614. SIAM, 2011.

\bibitem[Dud67]{dudley-chaining}
Richard~M Dudley.
\newblock The sizes of compact subsets of hilbert space and continuity of
  gaussian processes.
\newblock {\em Journal of Functional Analysis}, 1(3):290--330, 1967.

\bibitem[GH62]{gh-h-tum}
Alain Ghouila-Houri.
\newblock Caract{\'e}risation des matrices totalement unimodulaires.
\newblock {\em CR Acad. Sci. Paris}, 254:1192--1194, 1962.

\bibitem[LM12]{lovettmeka}
Shachar Lovett and Raghu Meka.
\newblock Constructive discrepancy minimization by walking on the edges.
\newblock In {\em Foundations of Computer Science (FOCS), 2012 IEEE 53rd Annual
  Symposium on}, pages 61--67. IEEE, 2012.

\bibitem[LSV86]{LSV}
L.~Lov{\'a}sz, J.~Spencer, and K.~Vesztergombi.
\newblock Discrepancy of set-systems and matrices.
\newblock {\em European Journal of Combinatorics}, 7(2):151--160, 1986.

\bibitem[Mat13]{Matousek11}
Ji{\v{r}}{\'\i} Matou{\v{s}}ek.
\newblock The determinant bound for discrepancy is almost tight.
\newblock {\em Proceedings of the American Mathematical Society},
  141(2):451--460, 2013.

\bibitem[Mat14]{matousek-says}
Jiri Matou{\v{s}}ek.
\newblock Personal communication, 2014.

\bibitem[MN12]{halfspaces}
S.~Muthukrishnan and Aleksandar Nikolov.
\newblock Optimal private halfspace counting via discrepancy.
\newblock In {\em STOC '12: Proceedings of the 44th symposium on Theory of
  Computing}, pages 1285--1292, New York, NY, USA, 2012. ACM.

\bibitem[Nik13]{komlos-sdp}
Aleksandar Nikolov.
\newblock The {K\'{o}mlos} conjecture holds for vector colorings.
\newblock {\em Under submission.}, 2013.

\bibitem[NTZ13]{NTZ}
Aleksandar Nikolov, Kunal Talwar, and Li~Zhang.
\newblock The geometry of differential privacy: the sparse and approximate
  cases.
\newblock In {\em Proceedings of the 45th annual ACM symposium on Symposium on
  theory of computing}, STOC '13, pages 351--360, New York, NY, USA, 2013. ACM.

\bibitem[Rot13]{rothvoss-binpacking}
Thomas Rothvo{\ss}.
\newblock Approximating bin packing within {O}(log {OPT} * log log {OPT}) bins.
\newblock In {\em FOCS}, pages 20--29, 2013.

\bibitem[Rot14]{Rothvoss14-giann}
Thomas Rothvo{\ss}.
\newblock Constructive discrepancy minimization for convex sets.
\newblock {\em CoRR}, abs/1404.0339, 2014.

\bibitem[Sch72]{schmidt}
Wolfgang Schmidt.
\newblock Irregularities of distribution, vii.
\newblock {\em Acta Arithmetica}, 21(1):45--50, 1972.

\bibitem[Sey80]{seymour-tum}
Paul~D Seymour.
\newblock Decomposition of regular matroids.
\newblock {\em Journal of combinatorial theory, Series B}, 28(3):305--359,
  1980.

\bibitem[Spe85]{spencer-six}
Joel Spencer.
\newblock Six standard deviations suffice.
\newblock {\em Transactions of the American Mathematical Society},
  289(2):679--706, 1985.

\bibitem[SS10]{bt-constructive}
D.A. Spielman and N.~Srivastava.
\newblock An elementary proof of the restricted invertibility theorem.
\newblock {\em Israel Journal of Mathematics}, pages 1--9, 2010.

\bibitem[Tal05]{talagrand-chaining}
Michel Talagrand.
\newblock {\em The Generic Chaining: Upper and Lower Bounds for Stochastic
  Processes}.
\newblock Springer, 2005.

\bibitem[Ver01]{vershynin}
R.~Vershynin.
\newblock John's decompositions: Selecting a large part.
\newblock {\em Israel Journal of Mathematics}, 122(1):253--277, 2001.

\end{thebibliography}

\appendix

\section{Weighted Restricted Invertibility}
\label{app:weighted-rip}

In this section we show that the weighted restricted invertibility
principle reduces to the standard version of the principle. Let us
first give a statement of the principle, in a version proved by
Spielman and Srivastava.

\begin{theorem}[\cite{bt-constructive}]\label{thm:bt-standard}
  Let $\epsilon > 0$, and let $A$ be an $m$ by $n$ real matrix. For
  any integer $k$ such that $k \leq \epsilon^2
  \frac{\|A\|_{HS}^2}{\|A\|_2^2}$ there exists a subset
  $S \subseteq [n]$ of size $S = k$ such that
  $\sigma_{\min}(A|_S)^2 \geq (1-\epsilon)^2\frac{\|A\|_{HS}^2}{n}$. 
  Moreover, $S$ can be computed in deterministic polynomial time. 
\end{theorem}

The reduction of Theorem~\ref{thm:bt} to Theorem~\ref{thm:bt-standard}
is based on the following simple technical lemma.


\begin{lemma}\label{lm:uniform}
  Let $Q\succeq 0$ be a diagonal matrix with rational entries, such that
  $\tr(Q) = 1$. Then for any $m$ by $n$ matrix
  $A = (a_j)_{j=1}^n$, there exists a $m \times n_Q$ matrix $A_Q$ such that
  $A_QA_Q^T = n_QAQA^T$. Moreover, all columns of $A_Q$ are columns of
  $A$.
\end{lemma}
\begin{proof}
  Let $n_Q$ be the least common denominator of all diagonal entries of $Q$,
  that is $n_QQ = R$ for  an integral diagonal matrix  $R$. Let
  then $A_Q$ be a matrix with $R_{jj}$ copies of each $a_j$. Clearly,
  $$
  A_QA_Q^T = \sum_{j=1}^n{R_{jj}a_ja_j^T} = ARA^T = n_Q AQA^T.
  $$
  Observe, finally, that the number of columns of $A_Q$ is equal to
  $\sum_{j=1}^n{R_{jj}} = n_Q\sum_{j=1}^n{Q_{jj}} = n_Q$. 
\end{proof}

\begin{proofof}{Theorem~\ref{thm:bt}}
    By introducing a tiny perturbation to $Q$, we can make it rational
    while changing $\|AQ^{1/2}\|_{HS}$ and $\|AQ^{1/2}\|_2$
    by an arbitrarily small amount. Therefore, we assume that $Q$ is
    rational. Then, by Lemma~\ref{lm:uniform}, there exists a matrix
    $A_Q$ with $n_Q$ columns all of which are columns of $A$, such
    that $A_QA_Q^T = n_Q AQA^T$. Let $k$ be arbitrary integer such
    that 
    \[
    k \leq \epsilon^2 \frac{\|AQ^{1/2}\|_{HS}^2}{\|AQ^{1/2}\|_2^2} 
    = \epsilon^2 \frac{\tr(AQA^T)}{\lambda_{\max}(AQA^T)}
    = \epsilon^2 \frac{n_Q\tr(A_QA_Q^T)}{n_Q\lambda_{\max}(A_QA_Q^T)}
    = \epsilon^2 \frac{\|A_Q\|_{HS}^2}{\|A_Q\|_2^2},
    \]
    where $\lambda_{\max}(M)$ is used to denote the largest eigenvalue
    of a matrix $M$. By Theorem~\ref{thm:bt-standard}, there exists a
    set $S_Q$ of size $k$, such that 
    \[
    \sigma_{\min}(A_Q|_{S_Q})^2 \geq   (1-\epsilon)^2\frac{\|A_Q\|_{HS}^2}{n_Q}
    = (1-\epsilon)^2 \frac{\tr(A_QA_Q^T)}{n_Q} = (1-\epsilon)^2 \tr(AQA^T) = (1-\epsilon)^2\|AQ^{1/2}\|_{HS}^2.
    \]
    But since all columns of $A_Q$ are also columns of $A$, and no
    column in $A_Q|_{S_Q}$ can be repeated or otherwise
    $\sigma_{\min}(A_Q|_{S_Q}) = 0$, there exists a set $S \subseteq
    [n]$ such that $\sigma_{\min}(A|_S)^2 \geq
    (1-\epsilon)^2\|AQ^{1/2}\|_{HS}^2$. 
\end{proofof}

\section{Vector Discrepancy for the Koml\'{o}s Problem}
\label{app:vector-komlos}

Theorem~\ref{thm:komlos-sdp} follows from the arguments
in~\cite{komlos-sdp} with few modifications. Here we sketch the
argument. 

We need a well-known lemma, also known as the Cauchy Interlace
Theorem. It follows easily from the variational characterization of
eigenvalues.
\begin{lemma}\label{lm:interlace}
  Let ${M}$ be a symmetric real matrix with eigenvalues $\lambda_1
  \geq \ldots \geq \lambda_n$. Let also ${U}\in \mathbb{R}^{n \times
    k}$ be a matrix with mutually orthogonal unit columns. Then for $1
  \leq i \leq k$, 
  $$
  \lambda_{n-k+i}(M) \leq \lambda_i({U^T} {M} {U}) \leq \lambda_{i}(M).
  $$
\end{lemma}

The following is an immediate consequence of
Lemma~\ref{lm:interlace}. 
\begin{lemma}\label{lm:interlace-det}
  Let ${M} \in \mathbb{R}^{n \times n}: {M} \succeq 0$ be a
  symmetric real matrix with eigenvalues $\sigma_1 \geq \ldots \geq \sigma_n \geq
  0$. Let also ${U}\in \mathbb{R}^{n \times k}$ be a matrix with
  mutually orthogonal unit columns. Then $\det({U^T} {M}
  {U}) \leq \sigma_1 \ldots \sigma_k$. 
\end{lemma}

The next lemma follows from strong duality for semidefinite
programming (see~\cite{Matousek11}). 
\begin{lemma}\label{lm:sdp-dual}
  For any real $m \times n$ matrix ${A}$, the minimum of
  $\|AXA^T\|_2^2$ over $X \succeq 0$ such that $\forall j \in [n]:
  X_{jj} = 1$ is equal to $D^2$ if and only if there exists a matrix $P
  \succeq 0$, $\tr(P) = 1$, and a diagonal matrix $Q\geq 0$, such that
  \begin{equation}\label{eq:cond-unif-w}
    \tr(Q) \geq D^2
  \end{equation}
 and for all  ${z} \in \mathbb{R}^n$ 
  \begin{equation} \label{eq:fact-unif}
    z^TA^TPAz \geq z^TQz.
  \end{equation}
\end{lemma}

The final lemma we need was proved in~\cite{komlos-sdp}. It can be
proved using the fact that the function $\sum_{i= 1}^n{e^{z_i}}$ is
symmetric and convex in $z = (z_i)_{i =1}^n$, and therefore is
Schur-convex. 
\begin{lemma}
  \label{lm:prod-sum}
  Let $x_1 \geq \ldots \geq x_n > 0$ and $y_1 \geq \ldots \geq y_n >
  0$ such that
  \begin{equation*}\label{eq:prod}
    \forall k \leq n: x_1 \ldots x_k \geq y_1 \ldots y_k
  \end{equation*}
  Then,
  \begin{equation*}\label{eq:sum}
    \forall k \leq n: x_1 + \ldots + x_k \geq y_1 + \ldots + y_k.
  \end{equation*}
\end{lemma}

\begin{theorem}[Theorem~\ref{thm:komlos-sdp}, restated]
  For any $m \times n$ matrix $A = (a_i)_{i = 1}^n$ satisfying
  $\forall i\in [n]: \|a_i\|_2 \leq 1$ there exists a $n\times n$
  matrix $X \succeq 0$ such that $\forall i \in [n]: X_{jj} = 1$ and
  $\|AXA^T\|_2 \leq 1$. 
\end{theorem}
\begin{proof}
    We will use Lemma~\ref{lm:sdp-dual} with $D =
  \sqrt{1+\epsilon}$ for an arbitrary $\epsilon > 0$.
  Assume for contradiction that there exist $P$ and $Q$ such that
  \eqref{eq:cond-unif-w} and \eqref{eq:fact-unif} are satisfied.
  Let us define $q_i = Q_{ii}$, and also  $p_i = \sigma_{i}(P)$. 
  Let, without loss of generality, $q_1 \geq \ldots \geq q_n > 0$. Denote by
  ${A_{[k]}}$ the matrix $(a_1, \ldots, a_k)$
  and by ${Q_{k}}$ the diagonal matrix with $q_1, \ldots, q_k$ on
  the diagonal. We first show that
  \begin{equation}
    \label{eq:eig-lb}
    \forall k \leq n: \det({A_{[k]}^T}{P}{A_{[k]}}) \leq
    p_1\ldots p_k.
  \end{equation}
  Let ${u_1}, \ldots {u_k}$ be an orthonormal basis for the
  range of ${A_{[k]}}$ and let ${U_k}$ be the matrix
  $({u_1}, \ldots {u_k})$. Then ${A_{[k]}} =
  {U_k}{U_k^T}{A_{[k]}}$. Each column of the square
  matrix ${U_k^T}{A_{[k]}}$ has norm at most $1$, and, by
  Hadamard's inequality, 
  \begin{equation*}
    \det({A^T_{[k]}U_k}) =  \det({U_k^T}{A_{[k]}}) \leq 1. 
  \end{equation*}
  Therefore, 
  \begin{equation*}
    \forall k \leq n: \det({A_{[k]}^T}{P}{A_{[k]}}) \leq
    \det({U_k^T}{P}{U_k}).
  \end{equation*}
  By Lemma~\ref{lm:interlace-det}, we have that
  $\det({U_k^T}{P}{U_k}) \leq p_1 \ldots p_k$, which
  proves (\ref{eq:eig-lb}). 

  By \eqref{eq:fact-unif} we know that for all $k$ and for all
  ${u} \in \mathbb{R}^k$, ${u^TA_{[k]}^TPA_{[k]}u} \geq
  {u^TQ_ku}$, since we can freely choose ${z}$ such that $z_i
  = 0$ for all $i>k$. Then, we have that
  \begin{equation}
    \label{eq:eig-ub}
    \forall k\leq n: \det({A_{[k]}^T}{P}{A_{[k]}}) \geq
    \det({Q_k}) = q_1\ldots q_k
  \end{equation}
  Combining (\ref{eq:eig-lb}) and (\ref{eq:eig-ub}), we have that 
  \begin{equation}
    \label{eq:p_k-w_k}
    \forall k \leq n: p_1 \ldots p_k \geq q_1 \ldots q_k
  \end{equation}
  By Lemma~\ref{lm:prod-sum}, (\ref{eq:p_k-w_k}) implies that $1 = \sum_{j = 1}^m{p_j} \geq
  \sum_{j = 1}^n{p_j} \geq \sum_{i = 1}^n{q_i} \geq 1+ \epsilon$, a
  contradiction. 

\end{proof}

\confoption{\section{Convex Duality}

We recall several basic facts from the theory of convex optimization,
which will be needed in the proof of Theorem~\ref{thm:nuclear} in Appendix~\ref{app:dual}. Assume we are given the following optimization problem:
\begin{align}
  &\text{Minimize } f_0(x)\label{eq:general-obj}\\
  &\text{s.t.}\notag\\
  &\forall 1\leq i \leq m: f_i(x) \leq 0.\label{eq:general-constr}
\end{align}
The Lagrange dual function associated with
\eqref{eq:general-obj}--\eqref{eq:general-constr} is defined as 
  $g(y) = \inf_x f_0(x) + \sum_{i = 1}^m{y_if_i(x)}$,
  where the infimum is over the intersection of the domains of
  $f_1,\ldots,\ldots f_m$, and $y \in \R^m$, $y \geq 0$. Since $g(y)$
  is the infimum of affine functions, it is a concave
  function. Moreover, $g$ is upper semi-continuous, and therefore
  continuous over the convex set $\{y: g(y) > -\infty\}$. 

  For any $x$ which is feasible for
  \eqref{eq:general-obj}--\eqref{eq:general-constr}, and any $y \geq
  0$, $g(y) \leq f_0(x)$. This fact is known as \emph{weak
    duality}. The \emph{Lagrange dual problem} is defined as
\begin{align}
  &\text{Maximize } g(y)
  \text{ s.t. }
  y \geq 0.\label{eq:L-dual}
\end{align}
\emph{Strong duality} holds when the optimal value of
\eqref{eq:L-dual} equals  the optimal
value of \eqref{eq:general-obj}--\eqref{eq:general-constr}. Slater's
condition is a commonly used sufficient condition for strong
duality. We state it next.

\begin{theorem}[Slater's Condition]\label{thm:slater}
  Assume $f_0, \ldots, f_m$ in the problem
  \eqref{eq:general-obj}--\eqref{eq:general-constr} are convex
  functions over their respective domains, and for some $k \geq 0$,
  $f_1, \ldots, f_k$ are affine functions. Let there be a point $x$ in
  the relative interior of the domains of $f_0, \ldots, f_m$, so that
  $f_i(x) \leq 0$ for $1 \leq i \leq k$ and $f_j(x) < 0$ for $k+1 \leq
  j \leq m$. Then the minimum of
  \eqref{eq:general-obj}--\eqref{eq:general-constr} equals the maximum
  of \eqref{eq:L-dual}, and the maximum of
  \eqref{eq:L-dual} is achieved if it is
  finite.
\end{theorem}

For more information on convex programming and duality, we refer the
reader to the book by Boyd and Vandenberghe~\cite{BoydV04-convexopt}.

\section{Proofs from Section~\ref{sect:min-ellips}}
\label{app:dual}

Recall that for a block
matrix \[U = \left(\begin{array}{cc}X &Y\\Y^T &Z\end{array}\right),\]
the \emph{Schur complement} of an invertible block $Z$ in $U$ is $X -
Y^TZ^{-1}Y$. When $Z \succ 0$, $U \succeq 0$ if and only if $X -
Y^TZ^{-1}Y \succeq 0$. This fact easily implies the convexity of the
matrix inverse.

\begin{lemma}\label{lm:inverse-convex}
  For any two $m\times m$ matrices $X \succ 0$ and $Y \succ 0$,
  $(\frac{1}{2}X + \frac{1}{2}Y)^{-1} \preceq \frac{1}{2}X^{-1} +
  \frac{1}{2}Y^{-1}$. 
\end{lemma}
\begin{proof}
  Define the matrices
  \begin{equation*}
    U = \left(
      \begin{array}{cc}
        X^{-1} &I\\
        I &X
      \end{array}
      \right)
      \;\;\;\;\;
      V = \left(
      \begin{array}{cc}
        Y^{-1} &I\\
        I &Y
      \end{array}
      \right).
  \end{equation*}
  The Schur complement of $X$ in $U$ is $0$, and
  therefore $U \succeq 0$, and analogously $V \succeq 0$. Therefore
  $U + V \succeq 0$, and the Schur complement of $X + Y$ in $U+V$ is also
  positive semidefinite, i.e.~$X^{-1} + Y^{-1} - 4(X+Y)^{-1} \succeq
  0$. This completes the proof, after re-arranging terms.
\end{proof}

\begin{proofof}{Lemma~\ref{lm:ellips-program}}
    Let $\lambda$ be the optimal value of
    \eqref{eq:ellips-obj}--\eqref{eq:ellips-enclose} and $\mu =
    \min\{\|E\|_\infty: \forall j \in [n]: a_j \in E\}$. Given a
    feasible $X$ for \eqref{eq:ellips-obj}--\eqref{eq:ellips-enclose},
    set $E = X^{-1/2}B_2^m$ (this is well-defined since $X \succ
    0$). Then for any $j \in [n]$, $\|a_j\|_E = a_j^TXa_j \leq 1$ by
    \eqref{eq:ellips-enclose}, and, therefore, $a_j \in E$. Also, by
    \eqref{eq:ellips-infty}, $\|E\|_\infty^2 = \max_{i = 1}^m
    e_i^TXe_i \leq t$. This shows that $\mu \leq \lambda$. In the
    reverse direction, let $E = FB_2^m$ be such that $\forall j\in
    [n]: a_j \in E$. Then, because $A$ is full rank, $F$ is also full
    rank and invertible, and we can define $X = (FF^T)^{-1}$ and $t =
    \|E\|_\infty^2$. Analogously to the calculations above, we can
    show that $X$ and $t$ are feasible, and therefore $\lambda \leq
    \mu$.

  The objective function and the constraints \eqref{eq:ellips-enclose}
  are affine, and therefore convex. To show \eqref{eq:ellips-width} are
  also convex, let $X_1, t_1$ and $X_2, t_2$ be two feasible solutions. Then,
  Lemma~\ref{lm:inverse-convex} implies that 
  for any $i$, $e_i^T(\frac{1}{2}X_1 + \frac{1}{2}X_2)^{-1}e_i \leq
  e_i^T\frac{1}{2}X_1^{-1}e_i + \frac{1}{2}e_i^TX_2^{-1}e_i \leq \frac{1}{2}t_1 +
  \frac{1}{2}t_2$, so constraints \eqref{eq:ellips-width} are convex as well.
\end{proofof}

\vspace{1em}
\begin{proofof}{Theorem~\ref{thm:nuclear}}
We shall prove the theorem by showing that the convex optimization
problem \eqref{eq:ellips-obj}--\eqref{eq:ellips-enclose} satisfies
Slater's condition, and that its Lagrange dual is equivalent to
\eqref{eq:nuclear-obj2}--\eqref{eq:nuclear-pos2}. Let us first
verify Slater's condition. We define the domain for constraints
\eqref{eq:ellips-width} as the open cone $\{X: X \succ 0\}$, which
makes the constraint $X \succ 0$ implicit. Let $X =
\frac{1}{\|A\|_{1\rightarrow 2}}I$, and $t = \|A\|_{1\rightarrow 2}+\varepsilon$
for some $\varepsilon > 0$. Then the affine constraints
\eqref{eq:ellips-enclose} are satisfied exactly, and the constraints
\eqref{eq:ellips-width} are satisfied with slack since $\varepsilon
> 0$. Moreover, by Lemma~\ref{lm:ellips-program}, all the
constraints and the objective function are convex. Therefore,
\eqref{eq:ellips-obj}--\eqref{eq:ellips-enclose} satisfies Slater's
condition, and consequently strong duality holds.

The Lagrange dual function for
\eqref{eq:ellips-obj}--\eqref{eq:ellips-enclose} is
\begin{equation*}
  g(p,r) = \inf_{t, X \succ 0}{t + \sum_{i = 1}^m{p_i(e_i^TX^{-1}e_i
      - t)} + \sum_{j = 1}^n{r_j(a_j^TXa_j - 1)} },
\end{equation*}
with dual variables $p \in \R^m$ and $r \in \R^n$, $p, r \geq
0$. Equivalently, writing $p$ as a diagonal matrix $P \in
\R^{m\times m}$, $P \succeq 0$, $r$ as a diagonal matrix $R \in
\R^{n\times n}$, $R \succeq 0$, we have $g(P,R) = \inf_{t,X\succ
  0}{t + \tr(PX^{-1}) - \tr(tP) + \tr(ARA^TX) - \tr(R)}$.  If $\tr(P)
\neq 1$, then $g(P,R) = -\infty$, since we can take $t$ to
$-\infty$ while keeping $X$ fixed. On the other hand, for $\tr(P)
= 1$, the dual function simplifies to
\begin{equation}\label{eq:g-raw}
  g(P,R) = \inf_{X \succ 0}{\tr(PX^{-1}) + \tr(ARA^TX) - \tr(R)}.
\end{equation}
Since $X\succ 0$ implies $X^{-1}\succ 0$, $g(P,R) \geq -\tr(R) >
-\infty$ whenever $\tr(P) = 1$. Therefore, $g(P,R)$ is continuous
over the set of diagonal positive semidefinite $P$, $R$ such that
$\tr(P) = 1$. For the rest of the proof we assume that $P$ and
$ARA^T$ are rank $m$. This is without loss of generality by the
continuity of $g$ and because both assumptions can be satisfied by
adding arbitrarily small perturbations to $P$ and $R$. (Here we use
the fact that $A$ is rank $m$.)

After differentiating the right hand side of \eqref{eq:g-raw} with
respect to $X$, we get the first-order optimality condition 
\begin{equation}
  \label{eq:fo-optimality}
  X^{-1}PX^{-1} =ARA^T.
\end{equation}
Multiplying by $P^{1/2}$ on the left and the right and taking square
roots gives the equivalent condition $P^{1/2}X^{-1}P^{1/2} =
(P^{1/2}ARA^TP^{1/2})^{1/2}$. This equation has a unique solution,
since $P$ and $ARA^T$ were both assumed to be invertible. Since
$\tr(PX^{-1}) = \tr(P^{1/2}X^{-1}P^{1/2})$ and also, by
\eqref{eq:fo-optimality}, $\tr(ARA^TX) = \tr(X^{-1}P) =
\tr(PX^{-1})$, we simplify $g(P,R)$ to
\begin{equation}\label{eq:g-final}
  g(P,R) = 2\tr((P^{1/2}ARA^TP^{1/2})^{1/2}) - \tr(R) =
  2\|P^{1/2}AR^{1/2}\|_{S_1} - \tr(R).
\end{equation}
We showed that
\eqref{eq:ellips-obj}--\eqref{eq:ellips-enclose} satisfies Slater's
condition and therefore strong duality holds, so by Theorem~\ref{thm:slater} and
Lemma~\ref{lm:ellips-program}, 
$\mu^2 = \max\{g(P,R): \tr(P) = 1, P,R \succeq 0, \text{ diagonal}\}$.
It remains to show that $g(P,R)$ is maximized when $\tr(R) =
\|P^{1/2}AR^{1/2}\|_{S_1}^2$. To this end, let us define new variables
$Q$ and $c$, where $c = \tr(R)$ and $Q = R/c$. Then we can re-write
$g(P,R)$ as
\begin{equation*}
  g(P,R) = g(P,Q,c) = 2\|P^{1/2}A(cQ)^{1/2}\|_{S_1} - \tr(cQ) =
  2\sqrt{c}\|P^{1/2}AQ^{1/2}\|_{S_1} - c. 
\end{equation*}
From the first-order optimality condition $\frac{dg}{dc} =
0$, we see that maximum of $g(P,Q,c)$ is achieved when $c =
\|P^{1/2}AQ^{1/2}\|_{S_1}^2$ and is equal to
$\|P^{1/2}AQ^{1/2}\|_{S_1}^2$. Therefore, maximizing $g(P,R)$ over
diagonal positive semidefinite $P$ and $R$ such that $\tr(P) = 1$ is
equivalent to the optimization problem
\eqref{eq:nuclear-obj2}--\eqref{eq:nuclear-pos2}. This completes the
proof. 
\end{proofof}}{}
\end{document}